\documentclass[aps,prl,twocolumn,nofootinbib,superscriptaddress,floatfix]{revtex4-2}

\usepackage[T1]{fontenc}
\usepackage{amsmath,amssymb,amsthm,mathtools,bm}
\usepackage{graphicx}
\usepackage{xcolor}
\usepackage{tikz}
\usetikzlibrary{arrows.meta,positioning,calc,fit}
\usepackage{microtype}
\usepackage[colorlinks=true,linkcolor=blue,citecolor=blue,urlcolor=blue]{hyperref}
\usepackage{enumitem}

\newcommand{\C}{\mathbb C}
\newcommand{\F}{\mathbb F}
\newcommand{\cH}{\mathcal H}
\newcommand{\cK}{\mathcal K}
\newcommand{\cB}{\mathcal B}
\newcommand{\cA}{\mathcal A}
\newcommand{\SEP}{\mathrm{SEP}}
\newcommand{\LOCC}{\mathrm{LOCC}}
\newcommand{\ind}{\mathrm{ind}}
\newcommand{\col}{\mathrm{col}}
\newcommand{\ad}{\mathrm{ad}}
\newcommand{\Tr}{\operatorname{Tr}}
\newcommand{\Sym}{\operatorname{Sym}}
\newcommand{\Stab}{\operatorname{Stab}}

\newcommand{\ket}[1]{|#1\rangle}
\newcommand{\bra}[1]{\langle#1|}

\newcommand{\proj}[1]{|#1\rangle\!\langle#1|}
\newcommand{\SMref}[1]{SM-\ref{#1}}

\newtheorem{theorem}{Theorem}
\newtheorem{proposition}{Proposition}
\newtheorem{lemma}{Lemma}
\newtheorem{corollary}{Corollary}
\theoremstyle{remark}

\begin{document}

\title{Power and Limits of Collective Local Measurements in Multicopy State Discrimination}

\author{Mao-Sheng Li}
\affiliation{School of Mathematics, South China University of Technology, Guangzhou 510641, China}
\author{Yan-Ling Wang}
\email{wangylmath@yahoo.com}
\affiliation{School of Computer Science and Technology, Dongguan University of Technology, Dongguan 523808, China}
 
\author{Zhu-Jun Zheng}
\email{zhengzj@scut.edu.cn}
\affiliation{School of Mathematics, South China University of Technology, Guangzhou 510641, China}

\date{September 6, 2026}

\begin{abstract}
More than two decades ago, Bennett \emph{et al.}
[\href{https://doi.org/10.1103/PhysRevA.59.1070}
{Phys. Rev. A \textbf{59}, 1070 (1999)}]
asked whether perfect local discrimination of orthogonal quantum states
can require more than two copies.  This question was subsequently
answered for adaptive protocols that process the copies separately
[\href{https://doi.org/10.1103/PhysRevLett.126.210505}
{Phys. Rev. Lett. \textbf{126}, 210505 (2021)}],
but remained open when each laboratory is allowed to process its local
copies collectively.  Here we resolve this stronger setting and identify
collective access across repeated local inputs as a distinct resource. For every fixed odd-prime local dimension, there exist complete maximal-stabilizer eigenbases whose copy complexity under individual-copy separable measurements diverges with system size.  Under collective processing this behavior changes sharply: every maximal-stabilizer eigenbasis in odd-prime local dimension is perfectly decoded by one-round collective LOCC using at most three copies.  Collective processing, however, does not remove
multicopy hardness in general.  For every fixed local dimension $d\ge2$, we prove the existence, within an explicit phase family, of complete bases whose copy complexity remains unbounded even under collective separable measurements, with a  square root of the number of subsystems as lower-bound scale.  Thus sample number, spatial measurement power, and coherent access across repeated local inputs are distinct resources in distributed quantum measurement.
\end{abstract}

\maketitle

 \paragraph{Introduction.---}
 The operational meaning of quantum-state distinguishability depends not
 only on the states, but also on the measurements that are physically
 available.  Mutually orthogonal states are perfectly distinguishable by
 a global measurement, yet spatial separation can make the same states
 locally indistinguishable.  This gap, first exposed by nonlocality
 without entanglement, has made local state discrimination a basic
 setting for studying the power and limitations of restricted quantum
 measurements~\cite{Bennett1999,BennettUPB1999,Walgate2000,
 	WalgateHardy2002,Ghosh2001,Horodecki2003}.  A broad theory has since
 emerged, including sharp bounds and separations for LOCC, separable
 (SEP), and positive-partial-transpose (PPT)
 measurements~\cite{Watrous2005,Hayashi2006,Duan2009,Yu2012,
 	Cosentino2013,Chitambar2014,Bandyopadhyay2015,
 	BandyopadhyayWalgate2009}, as well as structural criteria for
 entangled, maximally entangled, and mixed-state
 ensembles~\cite{Fan2004,Ghosh2004,Calsamiglia2010,Hashimoto2021,
 	LiShiWang2022,WangEtAl2025}.  In parallel, increasingly strong forms
 of local and genuine nonlocality have been identified for orthogonal
 product and entangled state
 sets~\cite{Halder2019,Rout2019,Bhattacharya2020,
 	BandyopadhyayHalder2021,RoutEtAl2021,ZhangStrong2019,
 	ZhouPlane2022,Zhou2023,HuGaoYan2024,HuGaoYan2025,
 	LiShiZhang2023,HeShiZhang2024,ZhenFeiZuo2022,LuCaoZuoFei2024,
 	ShiStrongUPB2022,CaoLiZuo2023,XiongLiZhengLi2023,LiZheng2022,
 	MurshidEtAl2026}.  These developments also clarify the role of
 auxiliary resources: shared entanglement can restore otherwise
 inaccessible measurements, while measurement incompatibility can itself
 become operationally useful~\cite{BandyopadhyayHalderNathanson2016,
 	BandyopadhyayHalderNathanson2018,SenHalderSen2024,
 	BandyopadhyayRusso2024}.  The same locality restrictions underlie
 locally accessible information~\cite{Badziag2003} and quantum data
 hiding~\cite{Terhal2001,DiVincenzoLeungTerhal2002,
 	EggelingWerner2002,Matthews2009}.
 
 The multicopy setting introduces a resource beyond the spatial
 measurement class: repeated preparations of the unknown state.
 Bennett \emph{et al.} asked whether perfect local discrimination of
 orthogonal states can require more than two copies~\cite{Bennett1999}.
 For pure states, Walgate \emph{et al.} showed that any two orthogonal
 states are perfectly distinguishable by LOCC~\cite{Walgate2000},
 which yields an $N-1$-copy upper bound for any set of $N$ orthogonal
 pure states by successive elimination.  Bandyopadhyay subsequently
 emphasized the finite-copy distinction between pure and mixed
 ensembles~\cite{Bandyopadhyay2011}, while Banik \emph{et al.} showed
 that even for a two-qubit orthonormal basis three adaptively processed
 copies can be necessary~\cite{Banik2021}.  Thus copy number is already
 a nontrivial resource for perfect local discrimination of orthogonal
 pure states.
 
    \begin{figure}[t]
 	\centering
 	\includegraphics[width=\columnwidth]{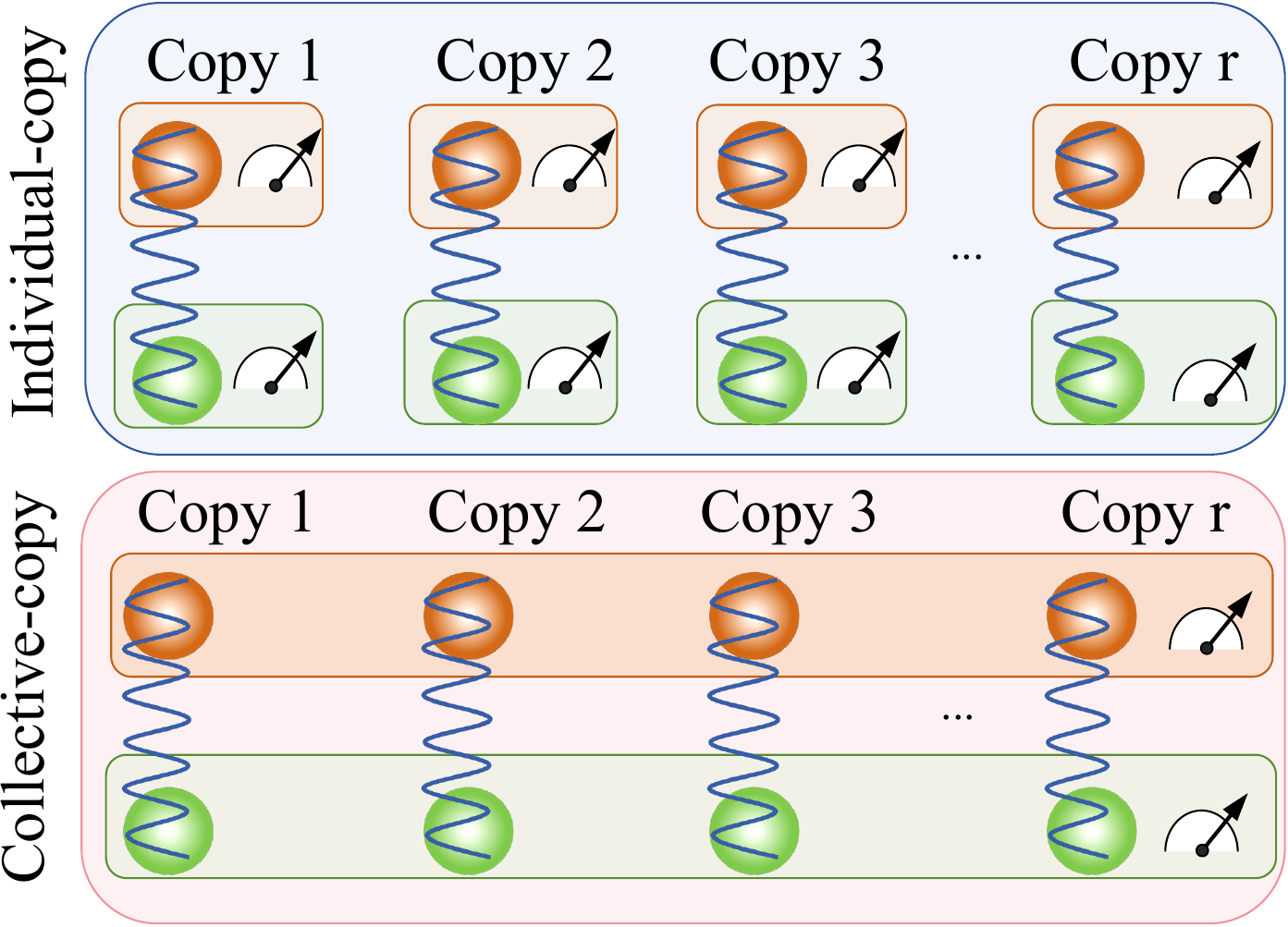}
 	\caption{Individual-copy and collective local measurements in multicopy
 		state discrimination.  In the individual-copy model (top), local
 		quantum operations do not couple distinct copies, although different
 		copies may be coordinated classically.  In the collective model
 		(bottom), laboratory $j$ may jointly process its local registers.   }
 	\label{fig:models}
 \end{figure}

 The same copies, however, need not be used in the same way.  As
 illustrated in Fig.~\ref{fig:models}, they may be processed separately,
 possibly with classical feed-forward, or collectively within each
 laboratory.  The latter enlarges the local quantum system on which a
 laboratory may act, without changing either the number of copies or the
 spatial measurement class. Collective advantages across copies are well established in the
 discrimination of nonorthogonal states
 ~\cite{PeresWootters1991, WalgateScott2008,Higgins2011,Jagannathan2022,
 	Conlon2023,TianEtAl2024,ConlonEtAl2025}.  For distributed orthogonal pure states, however, the role of collective
 local processing in multicopy discrimination has remained unclear.  This
leads to two complementary questions: can collective local measurements
change the \emph{scaling} of copy complexity, and how far can this
advantage extend?  In particular, does there exist a fixed number of
collectively processed copies that suffices for every complete
orthonormal basis?

 In this work, we answer these questions.  For every fixed odd-prime local dimension
 $p$, there exist complete maximal-stabilizer eigenbases whose copy
 complexity under individual-copy SEP measurements has a lower bound of
 order $\sqrt{n/\log n}$ and hence diverges with the number of parties $n$.
 Yet every maximal-stabilizer eigenbasis in odd-prime dimension is
 perfectly distinguished by one-round collective LOCC using at most
 three copies; for $p\equiv1\pmod4$, two copies suffice.  In that regime
 the separation is also information-theoretically maximal: two
 collective copies recover the full basis label, whereas individual-copy
 SEP below its copy-complexity scale reveals only a vanishing fraction
 of it.  Collective local measurements, however, are not universally
 sufficient.  For every fixed local dimension $d\ge2$, we construct
 complete bases within an explicit finite-field phase family whose copy
 complexity remains unbounded even under collective SEP measurements,
 with a $\sqrt n$ lower-bound scale and an exponential strong converse
 below it.  Thus copy number, spatial measurement power, and collective
 local processing across copies are distinct resources in multicopy
 state discrimination.

\paragraph{Measurement models and overlap bounds.---} 

Consider an $n$-partite system
\begin{equation}
 \cH=\bigotimes_{j=1}^{n}\cH_j,\qquad
 d_j=\dim\cH_j,\qquad
 D=\prod_{j=1}^{n}d_j,
 \label{eq:system}
\end{equation}
and a complete orthonormal basis $\cB=\{\ket{\psi_y}\}_{y=1}^{D}$, with $\rho_y=\proj{\psi_y}$. Throughout, pure states are represented by normalized state vectors unless explicitly stated otherwise. We use finite-round LOCC throughout; all LOCC achievability protocols below are in fact one-round. Besides LOCC, we use SEP as a relaxation for proving impossibility.

With several copies, the spatial measurement class does not by itself specify the available resources: one must also specify whether different copies may be processed jointly within each laboratory. In the \emph{individual-copy} architecture, no quantum operation couples distinct copies at a laboratory. In the \emph{collective} architecture, laboratory $j$ may instead act jointly on its entire local $r$-copy register,
\begin{equation}
 \cH^{\otimes r}\simeq\bigotimes_{j=1}^{n}\cH_j^{\otimes r}.
 \label{eq:groupedpartition}
\end{equation}
Here $\simeq$ denotes the canonical identification obtained by regrouping tensor factors according to laboratory.
Combining the copy architecture $\alpha\in\{\ind,\col\}$ with the spatial class $\mathsf M\in\{\LOCC,\SEP\}$ gives four measurement models. Adaptive LOCC~\cite{Banik2021}, in which copies are processed successively with classical feed-forward, is a subclass of individual-copy LOCC. The relevant inclusions are
\begin{equation}
\begin{array}{ccccc}
&&\SEP_{\ind}^{(r)}&\subseteq&\SEP_{\col}^{(r)}\\[-0.2em]
&&\rotatebox[origin=c]{90}{$\subseteq$}&&\rotatebox[origin=c]{90}{$\subseteq$}\\[-0.2em]
\LOCC_{\ad}^{(r)}&\subseteq&\LOCC_{\ind}^{(r)}&\subseteq&\LOCC_{\col}^{(r)} .
\end{array}
\label{eq:model-inclusions}
\end{equation}
Here $\SEP_{\ind}^{(r)}$ denotes the SEP relaxation that preserves the individual-copy structure; its precise definition and the inclusions above are proved in \SMref{app:trace}. In general, $\SEP_{\ind}^{(r)}$ and $\LOCC_{\col}^{(r)}$ are incomparable.

For a uniformly distributed basis label $Y\in\{1,\ldots,D\}$, let $\{M_y\}_{y=1}^D$ denote the $D$-outcome decision positive-operator-valued measure (POVM) used to guess $Y$. Define
\begin{align}
 P_{\alpha}^{\mathsf M}(r;\cB)
 &:=
 \sup_{\{M_y\}\in\mathsf M_{\alpha}^{(r)}}
 \frac1D\sum_y\Tr(M_y\rho_y^{\otimes r}),
 \label{eq:Pdef}\\
 N_{\alpha}^{\mathsf M}(\cB)
 &:=
 \min\{r\ge1:P_{\alpha}^{\mathsf M}(r;\cB)=1\}.
 \label{eq:PNdef}
\end{align}
Thus $N_{\alpha}^{\mathsf M}$ is the copy complexity of perfect discrimination in the corresponding model.

Lower bounds on $N_{\alpha}^{\mathsf M}$ follow from product overlaps. For a pure state $\ket{\psi}$, let
\begin{equation}
 \gamma(\ket{\psi})
 :=
 \max_{\substack{\ket{\phi}=\otimes_j\ket{\phi_j}\\ \langle\phi_j|\phi_j\rangle=1}}
 |\langle\phi|\psi\rangle|^2,
 \qquad
 \gamma(\cB):=\max_y\gamma(\ket{\psi_y}).
 \label{eq:gamma}
\end{equation}
For collective processing the product partition is different, since the $r$ systems held by one laboratory may be entangled with one another. We therefore define
\begin{equation}
 \Gamma_r(\ket{\psi})
 :=
 \max_{\substack{\ket{\Phi}=\otimes_j\ket{\Phi_j}\\ \langle\Phi_j|\Phi_j\rangle=1}}
 |\langle\Phi|\psi^{\otimes r}\rangle|^2,
 \qquad
 \ket{\Phi_j}\in\cH_j^{\otimes r},
 \label{eq:Gamma-state}
\end{equation}
and $\Gamma_r(\cB):=\max_y\Gamma_r(\ket{\psi_y})$.  The maxima in Eqs.~\eqref{eq:gamma} and \eqref{eq:Gamma-state} are attained because the relevant finite products of unit spheres are compact and the overlap functions are continuous.

\begin{proposition}[Trace-overlap bounds]\label{prop:trace-overlap}
For every complete orthonormal basis $\cB$ and every $r\ge1$,
\begin{equation}
 P_{\ind}^{\SEP}(r;\cB)\le D^{r-1}\gamma(\cB)^r,
 \label{eq:copytrace-main}
\end{equation}
and
\begin{equation}
 P_{\col}^{\SEP}(r;\cB)\le D^{r-1}\Gamma_r(\cB).
 \label{eq:collective-trace-main}
\end{equation}
Consequently, $D^{r-1}\gamma(\cB)^r<1$ implies $N_{\ind}^{\SEP}(\cB)>r$, whereas $D^{r-1}\Gamma_r(\cB)<1$ implies $N_{\col}^{\SEP}(\cB)>r$.
\end{proposition}

Proposition~\ref{prop:trace-overlap} converts optimization over separable measurements into geometric estimates of the intrinsic overlap parameters $\gamma(\cB)$ and $\Gamma_r(\cB)$.  The first bound accumulates the one-copy suppression as $\gamma(\cB)^r$, whereas the second is controlled by the genuinely grouped $r$-copy quantity $\Gamma_r(\cB)$; this distinction is the analytic bridge between the two copy architectures studied below.

The proof is given in \SMref{app:trace}. For example, if $D=d^n$ and $\gamma(\cB)=d^{-n+\Delta}$ with $\Delta>0$, then
\begin{equation}
 P_{\ind}^{\SEP}(r;\cB)\le d^{-n+r\Delta},\qquad
 N_{\ind}^{\SEP}(\cB)\ge\left\lceil\frac n\Delta\right\rceil .
 \label{eq:delta-to-N}
\end{equation}
Since $\gamma(\ket{\psi})\ge D^{-1}$ for every normalized $\ket{\psi}$, states with product overlap close to $D^{-1}$ are natural candidates for large individual-copy complexity.

Both overlap parameters are invariant under local unitaries. If a complete basis is a local-unitary orbit,
\begin{equation}
 \ket{\psi_y}=
 \left(\bigotimes_{j=1}^{n}U_{y,j}\right)\ket{\psi_0},
 \label{eq:LU-orbit-basis}
\end{equation}
then
\begin{equation}
 \gamma(\cB)=\gamma(\ket{\psi_0}),\qquad
 \Gamma_r(\cB)=\Gamma_r(\ket{\psi_0}).
 \label{eq:LU-overlap-invariance}
\end{equation}
Thus the relevant basis overlap is determined by a single representative state whenever its local-unitary orbit forms a complete orthonormal basis.

An especially transparent six-qubit benchmark is provided by absolutely maximally entangled (AME) states.  A pure $N$-party state of local dimension $d$ is $\operatorname{AME}(N,d)$ if every reduced state on at most $\lfloor N/2\rfloor$ parties is maximally mixed~\cite{Helwig2012}.  Let $\ket{\Psi}$ be any $\operatorname{AME}(6,2)$ state, choose any three-party subset $A$, and denote its complement by $\bar A$.  Write $\rho_A:=\Tr_{\bar A}\proj{\Psi}=I_8/8$, and let $\mathcal P_A$ denote the $64$-element tensor-product Pauli error basis on $A$.  The local-Pauli translates
\begin{equation}
 \cB_A(\Psi):=
 \left\{(P_A\otimes I_{\bar A})\ket{\Psi}:P_A\in\mathcal P_A\right\}
 \label{eq:AME-orbit-basis}
\end{equation}
form a complete orthonormal basis, because
\begin{equation}
 \langle\Psi|(P_A^\dagger Q_A\otimes I)|\Psi\rangle
 =\frac18\Tr(P_A^\dagger Q_A)=\delta_{P_A,Q_A}.
 \label{eq:AME-orbit-orthogonality}
\end{equation}
\SMref{app:ame} further proves the strict bound
\begin{equation}
 \gamma(\cB_A(\Psi))=\gamma(\ket{\Psi})<\frac18.
 \label{eq:AME-gamma}
\end{equation}
Hence Proposition~\ref{prop:trace-overlap} gives
\begin{equation}
 P_{\ind}^{\SEP}(2;\cB_A(\Psi))<1,\qquad
 N_{\ind}^{\SEP}(\cB_A(\Psi))\ge3.
 \label{eq:AME-N}
\end{equation}
Thus every $\operatorname{AME}(6,2)$ state generates a complete local-Pauli orbit basis that cannot be perfectly distinguished by two individual-copy SEP measurements; no special graph-state representative or flat computational-basis form is required.

To obtain scalable lower bounds, one must control the product overlap uniformly over all normalized product states. The following lemma converts a pointwise ensemble moment estimate into such a uniform bound.

\begin{lemma}[Moment-to-overlap uniformization]\label{lem:moment-to-overlap}
Let $\Theta$ be finite and let $\{\ket{\Psi_\theta}:\theta\in\Theta\}$ be normalized states on
\begin{equation}
 \cK=\bigotimes_{j=1}^{n}\C^{q_j}.
\end{equation}
Suppose that for some integer $m\ge1$ and $A_m\ge0$,
\begin{equation}
 \frac1{|\Theta|}\sum_{\theta\in\Theta}
 |\langle\Phi|\Psi_\theta\rangle|^{2m}\le A_m
 \label{eq:S-pointwise-family-moment}
\end{equation}
for every normalized product state $\ket{\Phi}=\otimes_j\ket{\phi_j}$ with $\langle\phi_j|\phi_j\rangle=1$ for all $j$. Then at least one $\theta_*$ satisfies
\begin{equation}
 \max_{\ket{\Phi}\ \mathrm{product}}
 |\langle\Phi|\Psi_{\theta_*}\rangle|^2
 \le
 \left[A_m\prod_{j=1}^{n}\binom{q_j+m-1}{m}\right]^{1/m}.
 \label{eq:S-moment-to-overlap}
\end{equation}
\end{lemma}

The lemma turns a high-moment estimate into the overlap quantities that control copy complexity. With $q_j=d_j$ it bounds $\gamma(\ket{\Psi_{\theta_*}})$, while after grouping $r$ copies and taking $q_j=d_j^r$ it bounds $\Gamma_r(\ket{\Psi_{\theta_*}})$. Combined with Proposition~\ref{prop:trace-overlap}, a suitable estimate of $A_m$ therefore yields lower bounds on either individual-copy or collective copy complexity. The proof is given in \SMref{app:uniformization}.

\paragraph{Stabilizer bases: unbounded individual-copy versus constant collective copy complexity.---} Let $p$ be an odd prime.  We use the standard $p$-ary Weyl--Pauli stabilizer formalism, reviewed explicitly in \SMref{app:stabmoment}.  A \emph{stabilizer state} is a normalized simultaneous eigenstate of a maximal Abelian subgroup of the generalized Pauli group, and a \emph{maximal-stabilizer eigenbasis} is the complete joint eigenbasis of such a subgroup.  Let $\Stab_{n,p}$ denote the finite set of stabilizer states on $(\C^p)^{\otimes n}$.  Their exact finite moments~\cite{Gross2006,KuengGross2015,Webb2016,Zhu2017,GrossNezamiWalter2021,Bittel2026} imply that, for every normalized product state $\ket{\beta}$ and $n\ge m-1$,
\begin{equation}
 \frac1{|\Stab_{n,p}|}
 \sum_{\ket{S}\in\Stab_{n,p}}
 |\langle\beta|S\rangle|^{2m}
 \le
 2^{m-1}p^{-nm+\binom{m-1}{2}}.
 \label{eq:stabilizer-Am-main}
\end{equation}
\SMref{app:stabmoment} derives Eq.~\eqref{eq:stabilizer-Am-main}, including the tensor-power matrix-element estimate needed to pass from the exact stabilizer moment identity to the product-state moment bound.

Lemma~\ref{lem:moment-to-overlap} then gives a stabilizer state $\ket{S_*}$ satisfying
\begin{equation}
 \gamma(\ket{S_*})
 \le
 2^{1-1/m}
 p^{-n+(m-1)(m-2)/(2m)}
 \binom{p+m-1}{m}^{n/m}.
 \label{eq:stab-gamma-main}
\end{equation}
Optimizing $m$ yields the first main result.

\begin{theorem}[Unbounded individual-copy complexity]\label{thm:individual-unbounded}
For every fixed odd prime $p$, there exists a sequence of complete maximal-stabilizer eigenbases $\{\cB_n\}$ of $(\C^p)^{\otimes n}$ such that
\begin{equation}
 N_{\ind}^{\SEP}(\cB_n)
 \ge
 \left(\frac1{\sqrt{p-1}}-o(1)\right)
 \sqrt{\frac{n}{\log_p n}}
 \longrightarrow\infty.
 \label{eq:thm1-N}
\end{equation}
Unless another limit is stated explicitly, all asymptotic statements refer to $n\to\infty$ with the physical local dimensions fixed.
Moreover, if $r=o(\sqrt{n/\log_p n})$, then
\begin{equation}
 P_{\ind}^{\SEP}(r;\cB_n)\le D^{-1+o(1)}.
 \label{eq:thm1-strong}
\end{equation}
\end{theorem}

The second statement shows that below the threshold the identification probability remains within a subexponential factor of random guessing. The stabilizer-moment estimate and its asymptotic optimization are proved in \SMref{app:stabmoment}.

Collective processing changes this behavior completely. Although the stabilizer generators commute globally, their local Weyl factors need not commute within a single copy, so their eigenvalue information cannot in general be read out simultaneously by local measurements. Access to several local repetitions changes this compatibility structure: suitable collective observables can be made mutually commuting at every laboratory while retaining enough global eigenvalue information to reconstruct the stabilizer label. This gives the following uniform bound.

\begin{theorem}[Collective stabilizer readout]\label{thm:stabilizer-collective}
Let $\cB_{\mathrm{stab}}$ be any maximal-stabilizer eigenbasis in odd-prime local dimension $p$. Then $\cB_{\mathrm{stab}}$ is perfectly distinguishable by one-round collective LOCC using at most
\begin{equation}
 \tau_p=
 \begin{cases}
 2,&p\equiv1\pmod4,\\
 3,&p\equiv3\pmod4
 \end{cases}
 \label{eq:taup}
\end{equation}
copies. Hence
\begin{equation}
 N_{\col}^{\SEP}(\cB_{\mathrm{stab}})
 \le
 N_{\col}^{\LOCC}(\cB_{\mathrm{stab}})
 \le\tau_p\le3.
\end{equation}
\end{theorem}

The protocol requires only one round of classical communication and is independent of the number of laboratories. In particular, two copies suffice in local dimension five. The collective observables, syndrome reconstruction, and the finite-field choice of two or three scalar weights are given in \SMref{app:stabilizerreadout}.

Theorems~\ref{thm:individual-unbounded} and~\ref{thm:stabilizer-collective} give an unbounded separation within the same maximal-stabilizer family: the individual-copy requirement diverges, whereas collective LOCC needs at most three copies. For $p=5$, two collective copies suffice. The separation also appears directly in the recoverable classical information. Let $Y$ be the uniformly distributed basis label. For an allowed measurement $\mathcal M\in\mathsf M_{\alpha}^{(r)}$ with classical outcome $Z_{\mathcal M}$, let $I(Y{:}Z_{\mathcal M})$ denote their base-$2$ Shannon mutual information, and define
\begin{equation}
 I_{\alpha}^{\mathsf M}(r;\cB):=\sup_{\mathcal M\in\mathsf M_{\alpha}^{(r)}} I(Y:Z_{\mathcal M}).
 \label{eq:Idef}
\end{equation}
Then
\begin{equation}
 I_{\alpha}^{\mathsf M}(r;\cB)\le\log_2D+\log_2P_{\alpha}^{\mathsf M}(r;\cB),
 \label{eq:info-vs-P-main}
\end{equation}
as proved in \SMref{app:information}. Therefore:

\begin{corollary}[Asymptotically maximal information gap]\label{cor:information-separation}
For every fixed odd prime $p\equiv1\pmod4$, the bases in Theorem~\ref{thm:individual-unbounded} can be chosen so that
\begin{equation}
 I_{\col}^{\LOCC}(2;\cB_n)=\log_2D,\qquad
 I_{\ind}^{\SEP}(r;\cB_n)=o(\log_2D)
 \label{eq:info-gap-compact}
\end{equation}
for every $r=o(\sqrt{n/\log_p n})$. In particular,
\begin{equation}
 I_{\col}^{\LOCC}(2;\cB_n)-I_{\ind}^{\SEP}(2;\cB_n)
 =(1-o(1))\log_2D.
 \label{eq:info-gap}
\end{equation}
\end{corollary}

Thefore,  the same two preparations change the readable information from $o(n)$ bits to the full $n\log_2 p$ bits, without quantum communication between laboratories.

\paragraph{Unbounded collective complexity.---}
The stabilizer decoder motivates the converse problem of whether a fixed collective block can suffice for every complete basis. We show that it cannot by using flat-amplitude states whose orbit orthogonality is automatic, leaving the phases available to suppress collective product overlaps.

For fixed $d\ge2$, write $\mathbb Z_d:=\mathbb Z/d\mathbb Z$ and set $D=d^n$.  For a global computational string $\mathbf{x}=(x_1,\ldots,x_n)\in\mathbb Z_d^n$, write $\ket{\mathbf{x}}:=\bigotimes_{j=1}^n\ket{x_j}$.  Given a phase function $f:\mathbb Z_d^n\to\mathbb R$, consider
\begin{equation}
\ket{\Omega_f}=D^{-1/2}\sum_{\mathbf{x}\in\mathbb Z_d^n}e^{if(\mathbf{x})}\ket{\mathbf{x}}.
\label{eq:flat-general}
\end{equation}
Let $\zeta=e^{2\pi i/d}$ and define the one-site phase operator $Z_d\ket a=\zeta^a\ket a$ for $a\in\mathbb Z_d$.  For $\mathbf{y}=(y_1,\ldots,y_n)\in\mathbb Z_d^n$, set $Z_d^{\mathbf{y}}:=\bigotimes_{j=1}^n Z_d^{y_j}$ and $\mathbf{y}\!\cdot\!\mathbf{x}:=\sum_{j=1}^ny_jx_j$ modulo $d$. Then $Z_d^{\mathbf{y}}\ket{\mathbf{x}}=\zeta^{\mathbf{y}\cdot\mathbf{x}}\ket{\mathbf{x}}$, and for any $\mathbf{y},\mathbf{z}\in\mathbb Z_d^n$,
\begin{equation}
\langle\Omega_f|Z_d^{\mathbf{z}-\mathbf{y}}|\Omega_f\rangle=D^{-1}\sum_{\mathbf{x}\in\mathbb Z_d^n}\zeta^{(\mathbf{z}-\mathbf{y})\cdot\mathbf{x}}=\delta_{\mathbf{y},\mathbf{z}}.
\label{eq:flat-orthogonality}
\end{equation}
where $\delta_{\mathbf y,\mathbf z}$ is the Kronecker delta. Therefore the local-unitary orbit
\begin{equation}
\cA_f=\{Z_d^{\mathbf{y}}\ket{\Omega_f}:\mathbf{y}\in\mathbb Z_d^n\}
\label{eq:flat-orbit}
\end{equation}
is a complete orthonormal basis. Flat amplitudes therefore decouple two tasks: the orbit guarantees completeness, while the phases can be optimized solely against collective product detectors.

Choose the phases from the explicit finite-field moment-curve family developed in \SMref{app:flat}. For integers $r,m\ge1$, set $h:=mr$. Using an auxiliary prime $P>\max\{D,h\}$ and writing $\F_P$ for the field with $P$ elements, assign distinct nonzero field points $\alpha_{\mathbf x}\in\F_P$ to the $D$ computational labels and define
\begin{equation}
 \mathbf g_{\mathbf x}:=(\alpha_{\mathbf x},\alpha_{\mathbf x}^2,\ldots,\alpha_{\mathbf x}^h)^T\in\F_P^h.
 \label{eq:momentcurve-vector-main}
\end{equation}
For $\mathbf u\in\F_P^h$, set
\begin{equation}
 \ket{\Omega_{\mathbf u}}:=D^{-1/2}\sum_{\mathbf x\in\mathbb Z_d^n}\omega_P^{\mathbf u\cdot\mathbf g_{\mathbf x}}\ket{\mathbf x},\qquad \omega_P:=e^{2\pi i/P}.
 \label{eq:momentcurve-main}
\end{equation}
Newton identities show that the first $h$ power sums determine the multiset of $h$ field labels.  Combined with character orthogonality over $\F_P^h$, this replica-rigidity property gives for every normalized product state $\ket{\Phi}=\bigotimes_j\ket{\Phi_j}$ across the grouped laboratory partition with $\ket{\Phi_j}\in(\C^d)^{\otimes r}$ and $\langle\Phi_j|\Phi_j\rangle=1$,
\begin{equation}
\frac1{P^h}\sum_{\mathbf{u}\in\F_P^h}|\langle\Phi|\Omega_{\mathbf{u}}^{\otimes r}\rangle|^{2m}\le\frac{h!}{D^h}=\frac{(mr)!}{D^{mr}}.
\label{eq:phase-Am-main}
\end{equation}
Applying Lemma~\ref{lem:moment-to-overlap} with grouped local dimension $q_j=d^r$ yields a state $\ket{\Omega_{\mathbf{u}_*}}$ satisfying
\begin{equation}
 \Gamma_r(\ket{\Omega_{\mathbf{u}_*}})\le D^{-r}\binom{d^r+m-1}{m}^{n/m}(mr)!^{1/m}.
 \label{eq:flat-Gamma-main}
\end{equation}
Define $\cA_{\mathbf{u}_*,d}:=\{Z_d^{\mathbf{y}}\ket{\Omega_{\mathbf{u}_*}}:\mathbf{y}\in\mathbb Z_d^n\}$. Since this is a local-unitary orbit of $\ket{\Omega_{\mathbf{u}_*}}$,
\begin{equation}
 \Gamma_r(\cA_{\mathbf{u}_*,d})=\Gamma_r(\ket{\Omega_{\mathbf{u}_*}}).
\end{equation}
Proposition~\ref{prop:trace-overlap} therefore gives
\begin{equation}
 P_{\col}^{\SEP}(r;\cA_{\mathbf{u}_*,d})\le\frac1D\binom{d^r+m-1}{m}^{n/m}(mr)!^{1/m}.
 \label{eq:flat-master-main}
\end{equation}
The same analytic bridge used for stabilizer states therefore controls the larger product manifold relevant to collective processing.  SM-\ref{app:phasemoment} embeds the fixed-$r$ moment curves into a common finite-field construction and proves that one phase parameter can satisfy the required bounds throughout the block-size range used below.

\begin{theorem}[Unbounded collective copy complexity]\label{thm:collective-unbounded}
For every fixed integer $d\ge2$, there exists a sequence of complete orthonormal bases $\{\cA_{n,d}\}$ of $(\C^d)^{\otimes n}$ such that
\begin{equation}
N_{\col}^{\SEP}(\cA_{n,d})\ge\sqrt n-\frac12\log_d n-\frac12\log_d\ln n-O_d(1).
\label{eq:thm3-N}
\end{equation}
Moreover, for every fixed $0<\alpha<1$ and $r=\lfloor\alpha\sqrt n\rfloor$,
\begin{equation}
P_{\col}^{\SEP}(r;\cA_{n,d})\le D^{-(1-\alpha^2)+o(1)},
\label{eq:thm3-strong}
\end{equation}
and hence $I_{\col}^{\SEP}(r;\cA_{n,d})\le(\alpha^2+o(1))\log_2D$.
\end{theorem}
In particular, for every prescribed integer $R$ and every fixed local dimension $d\ge2$, there exists a complete basis with $N_{\col}^{\SEP}\ge R$. The physical local dimension need not be prime: the auxiliary prime used in the moment-curve phases is unrelated to the physical alphabet, whose orbit orthogonality uses only the characters of $\mathbb Z_d^n$.

\vskip 5pt

\paragraph{Operational implications and applications.---}
The separation above defines a readout primitive controlled by the local measurement architecture.  Suppose that a classical
label $y$ is encoded into a distributed basis state $\ket{\psi_y}$ and
the subsystems are delivered to spatially separated laboratories.
The communication pattern between the laboratories can be kept fixed,
while access to the label is controlled solely by the local ability to
retain and jointly process repeated inputs.  For the $p=5$ stabilizer
family of Theorem~\ref{thm:individual-unbounded}, two repetitions are
sufficient to recover the entire $n\log_2 5$-bit label by one-round
collective LOCC, whereas any individual-copy SEP measurement in the
subthreshold regime reveals only $o(n)$ bits.  Thus local quantum
memory and inter-copy gates can unlock an extensive classical message
without introducing quantum communication between the nodes.

This gives a form of measurement-restricted data hiding in which the
access restriction is temporal rather than spatial.  Conventional
data-hiding schemes protect information by limiting the nonlocal
operations available between separated parties.  Here the spatial
measurement class may remain unchanged: information is instead hidden
from devices that cannot maintain quantum coherence across repeated
local inputs and becomes accessible once that capability is supplied.
 The same construction therefore distinguishes two physically different resources—quantum communication between laboratories and coherent storage and joint processing of repeated inputs within each laboratory—that are often treated together in a generic many-copy measurement model.

The stabilizer decoder also applies naturally to distributed
stabilizer readout.  Stabilizer eigenvalues provide classical syndrome
labels in quantum error correction, graph-state processing, and
networked stabilizer protocols.  Our result shows that local
incompatibility of the stabilizer factors need not force either quantum
communication or a growing number of separately processed
preparations: at most three jointly stored repetitions suffice to make
the relevant local observables simultaneously measurable in every
odd-prime dimension.  Theorem~\ref{thm:collective-unbounded}, however,
shows that this compression is structure dependent.  No fixed local
memory block can provide universal distributed readout for arbitrary
orthogonal encodings.

\paragraph{Discussion.---}
Our results show that copy number alone is not an operationally complete
description of multicopy local measurement.  The same repeated
preparations can have fundamentally different discrimination power
depending on whether each laboratory can maintain coherence across its
local copies and process them jointly.  This distinction is asymptotic,
rather than a finite-size improvement: for the stabilizer families of
Theorem~\ref{thm:individual-unbounded}, the number of separately
processed copies required by SEP diverges with system size, whereas
collective LOCC requires at most three copies for every system size.
For $p\equiv1\pmod4$, the separation is already complete with two
copies: collective LOCC recovers the full $\log_2 D$-bit basis label,
while individual-copy SEP in the subthreshold regime accesses only
$o(\log D)$ bits.  Thus the relevant resource is not an additional
sample or an additional communication link, but the ability to preserve
and exploit quantum coherence between repeated inputs held at the same
laboratory.  In this sense, local inter-copy coherence constitutes a
measurement resource distinct from both sample number and spatial
nonlocality.

The contrast between Theorems~\ref{thm:stabilizer-collective} and
\ref{thm:collective-unbounded} also shows that this resource is strongly
structure dependent.  For maximal-stabilizer eigenbases, collective
copies remove the local incompatibility of globally commuting
observables: the scalar-weight construction cancels all local Weyl
commutators simultaneously while preserving enough global eigenvalue
information to reconstruct the complete stabilizer syndrome.  The
flat-phase family exhibits the opposite behavior.  Even after arbitrary
joint processing of all copies stored at each laboratory, the
collective-SEP copy complexity remains unbounded, and below the
$\sqrt n$ scale the optimal discrimination probability obeys an
exponential strong converse.  Taken together, these results indicate
that bounded collective copy complexity is not determined simply by
entanglement, orthogonality, or the size of the state set.  Rather, it
appears to reflect whether the distributed operator structure of the
basis admits a finite-copy compatibility mechanism.  This provides a
possible structural principle for classifying multicopy local
measurement problems beyond the stabilizer setting.

Several questions become natural from this viewpoint.  First, the
optimal worst-case scaling of collective copy complexity is unknown:
our construction establishes an order-$\sqrt n$ lower-bound scale, but
it remains open whether substantially larger growth, possibly linear
in $n$, can occur for complete bases of fixed local dimension.
Second, it would be valuable to characterize those bases for which a
bounded number of collective copies suffices, and to determine whether
the scalar-weight mechanism is a special feature of Weyl--Pauli
stabilizers or an instance of a more general compatibility principle.
Even within the stabilizer setting, it is open whether two collective
copies always suffice when $p\equiv3\pmod4$, or whether three copies are
genuinely necessary for some bases.  Finally, the present results concern
perfect discrimination in an idealized measurement model.  A natural
next step is to quantify the tradeoff among discrimination error, the
number of copies, and the amount of coherent local storage under noise
or finite-memory constraints.  Such a theory would turn collective
copy complexity from a discrimination parameter into an operational
measure of local quantum-memory requirements in distributed quantum
information processing.

\vskip 5pt

\begin{acknowledgments}
\noindent \emph {Acknowledgments}: This work is supported by the National Natural Science Foundation of China under Grant No.~12371458 and the Guangdong Basic and Applied Basic Research Foundation under Grants Nos.~2024A1515030023 and 2024A1515010380.
\end{acknowledgments}

\vskip 10pt
\emph{AI statement.—}
We used OpenAI's GPT-5.6 primarily to assist with language polishing,
clarifying the presentation, and editing the final manuscript.
In addition, GPT-5.6 provided a proof of Lemma~\ref{lem:rT-contraction}
and suggested the accessible-information inequality stated in
Lemma~\ref{lem:information-identification}, together with a preliminary
proof that was subsequently checked, completed, and rewritten in detail
by the authors.
All other ideas, constructions, main results, and proofs in this work
were developed independently by the authors.

\bibliographystyle{apsrev4-2}
 
\bibliography{MD_ref}

@article{Badziag2003,
  author  = {Badzi{\k a}g, P. and Horodecki, M. and Sen(De), A. and Sen, U.},
  title   = {{Locally Accessible Information: How Much Can the Parties Gain by Cooperating?}},
  journal = {Phys. Rev. Lett.},
  volume  = {91},
  pages   = {117901},
  year    = {2003},
  doi     = {10.1103/PhysRevLett.91.117901},
  url     = {https://doi.org/10.1103/PhysRevLett.91.117901}
}

@article{Bandyopadhyay2011,
  author  = {Bandyopadhyay, S.},
  title   = {{More Nonlocality with Less Purity}},
  journal = {Phys. Rev. Lett.},
  volume  = {106},
  pages   = {210402},
  year    = {2011},
  doi     = {10.1103/PhysRevLett.106.210402},
  url     = {https://doi.org/10.1103/PhysRevLett.106.210402}
}

@article{Bandyopadhyay2015,
  author  = {Bandyopadhyay, S. and Cosentino, A. and Johnston, N. and Russo, V. and Watrous, J. and Yu, N.},
  title   = {{Limitations on Separable Measurements by Convex Optimization}},
  journal = {IEEE Trans. Inf. Theory},
  volume  = {61},
  pages   = {3593--3604},
  year    = {2015},
  doi     = {10.1109/TIT.2015.2417755},
  url     = {https://doi.org/10.1109/TIT.2015.2417755}
}

@article{BandyopadhyayHalder2021,
  author  = {Bandyopadhyay, S. and Halder, S.},
  title   = {{Genuine Activation of Nonlocality: From Locally Available to Locally Hidden Information}},
  journal = {Phys. Rev. A},
  volume  = {104},
  pages   = {L050201},
  year    = {2021},
  doi     = {10.1103/PhysRevA.104.L050201},
  url     = {https://doi.org/10.1103/PhysRevA.104.L050201}
}

@article{BandyopadhyayHalderNathanson2016,
  author  = {Bandyopadhyay, S. and Halder, S. and Nathanson, M.},
  title   = {{Entanglement as a Resource for Local State Discrimination in Multipartite Systems}},
  journal = {Phys. Rev. A},
  volume  = {94},
  pages   = {022311},
  year    = {2016},
  doi     = {10.1103/PhysRevA.94.022311},
  url     = {https://doi.org/10.1103/PhysRevA.94.022311}
}

@article{BandyopadhyayHalderNathanson2018,
  author  = {Bandyopadhyay, S. and Halder, S. and Nathanson, M.},
  title   = {{Optimal Resource States for Local State Discrimination}},
  journal = {Phys. Rev. A},
  volume  = {97},
  pages   = {022314},
  year    = {2018},
  doi     = {10.1103/PhysRevA.97.022314},
  url     = {https://doi.org/10.1103/PhysRevA.97.022314}
}

@article{BandyopadhyayRusso2024,
  author  = {Bandyopadhyay, S. and Russo, V.},
  title   = {{Locally Distinguishing a Maximally Entangled Basis Using Shared Entanglement}},
  journal = {Phys. Rev. A},
  volume  = {110},
  pages   = {042406},
  year    = {2024},
  doi     = {10.1103/PhysRevA.110.042406},
  url     = {https://doi.org/10.1103/PhysRevA.110.042406}
}

@article{BandyopadhyayWalgate2009,
  author  = {Bandyopadhyay, S. and Walgate, J.},
  title   = {{Local Distinguishability of Any Three Quantum States}},
  journal = {J. Phys. A: Math. Theor.},
  volume  = {42},
  pages   = {072002},
  year    = {2009},
  doi     = {10.1088/1751-8113/42/7/072002},
  url     = {https://doi.org/10.1088/1751-8113/42/7/072002}
}

@article{Banik2021,
  author  = {Banik, M. and Guha, T. and Alimuddin, M. and Kar, G. and Halder, S. and Bhattacharya, S. S.},
  title   = {{Multicopy Adaptive Local Discrimination: Strongest Possible Two-Qubit Nonlocal Bases}},
  journal = {Phys. Rev. Lett.},
  volume  = {126},
  pages   = {210505},
  year    = {2021},
  doi     = {10.1103/PhysRevLett.126.210505},
  url     = {https://doi.org/10.1103/PhysRevLett.126.210505}
}

@article{Bennett1999,
  author  = {Bennett, C. H. and DiVincenzo, D. P. and Fuchs, C. A. and Mor, T. and Rains, E. and Shor, P. W. and Smolin, J. A. and Wootters, W. K.},
  title   = {{Quantum Nonlocality without Entanglement}},
  journal = {Phys. Rev. A},
  volume  = {59},
  pages   = {1070--1091},
  year    = {1999},
  doi     = {10.1103/PhysRevA.59.1070},
  url     = {https://doi.org/10.1103/PhysRevA.59.1070}
}

@article{BennettUPB1999,
  author  = {Bennett, C. H. and DiVincenzo, D. P. and Mor, T. and Shor, P. W. and Smolin, J. A. and Terhal, B. M.},
  title   = {{Unextendible Product Bases and Bound Entanglement}},
  journal = {Phys. Rev. Lett.},
  volume  = {82},
  pages   = {5385--5388},
  year    = {1999},
  doi     = {10.1103/PhysRevLett.82.5385},
  url     = {https://doi.org/10.1103/PhysRevLett.82.5385}
}

@article{Bhattacharya2020,
  author  = {Bhattacharya, S. S. and Saha, S. and Guha, T. and Banik, M.},
  title   = {{Nonlocality without Entanglement: Quantum Theory and Beyond}},
  journal = {Phys. Rev. Research},
  volume  = {2},
  pages   = {012068(R)},
  year    = {2020},
  doi     = {10.1103/PhysRevResearch.2.012068},
  url     = {https://doi.org/10.1103/PhysRevResearch.2.012068}
}

@article{Bittel2026,
  author = {Bittel, L. and Eisert, J. and Leone, L. and Mele, A. A. and Oliviero, S. F. E.},
  title = {A complete theory of the {Clifford} commutant},
  journal = {Quantum}, volume = {10}, pages = {2171}, year = {2026},
  doi = {10.22331/q-2026-07-22-2171}, url = {https://doi.org/10.22331/q-2026-07-22-2171}
}

@article{Calsamiglia2010,
  author  = {Calsamiglia, J. and de Vicente, J. I. and Mu{\~n}oz-Tapia, R. and Bagan, E.},
  title   = {{Local Discrimination of Mixed States}},
  journal = {Phys. Rev. Lett.},
  volume  = {105},
  pages   = {080504},
  year    = {2010},
  doi     = {10.1103/PhysRevLett.105.080504},
  url     = {https://doi.org/10.1103/PhysRevLett.105.080504}
}

@article{CaoLiZuo2023,
  author  = {Cao, H.-Q. and Li, M.-S. and Zuo, H.-J.},
  title   = {{Locally stable sets with minimum cardinality}},
  journal = {Phys. Rev. A},
  volume  = {108},
  pages   = {012418},
  year    = {2023},
  doi     = {10.1103/PhysRevA.108.012418},
  url     = {https://doi.org/10.1103/PhysRevA.108.012418}
}

@article{Chitambar2014,
  author  = {Chitambar, E. and Leung, D. and Man{\v c}inska, L. and Ozols, M. and Winter, A.},
  title   = {{Everything You Always Wanted to Know About {LOCC} (But Were Afraid to Ask)}},
  journal = {Commun. Math. Phys.},
  volume  = {328},
  pages   = {303--326},
  year    = {2014},
  doi     = {10.1007/s00220-014-1953-9},
  url     = {https://doi.org/10.1007/s00220-014-1953-9}
}

@article{Conlon2023,
  author  = {Conlon, L. O. and Eilenberger, F. and Lam, P. K. and Assad, S. M.},
  title   = {{Discriminating Mixed Qubit States with Collective Measurements}},
  journal = {Commun. Phys.},
  volume  = {6},
  pages   = {337},
  year    = {2023},
  doi     = {10.1038/s42005-023-01454-z},
  url     = {https://doi.org/10.1038/s42005-023-01454-z}
}

@article{ConlonEtAl2025,
  author  = {Conlon, L. O. and Koh, J. M. and Shajilal, B. and Sidhu, J. and Lam, P. K. and Assad, S. M.},
  title   = {{Attainability of Quantum State Discrimination Bounds with Collective Measurements on Finite Copies}},
  journal = {Phys. Rev. A},
  volume  = {111},
  pages   = {022438},
  year    = {2025},
  doi     = {10.1103/PhysRevA.111.022438},
  url     = {https://doi.org/10.1103/PhysRevA.111.022438}
}

@article{Cosentino2013,
  author  = {Cosentino, A.},
  title   = {{Positive-Partial-Transpose-Indistinguishable States via Semidefinite Programming}},
  journal = {Phys. Rev. A},
  volume  = {87},
  pages   = {012321},
  year    = {2013},
  doi     = {10.1103/PhysRevA.87.012321},
  url     = {https://doi.org/10.1103/PhysRevA.87.012321}
}

@article{DiVincenzoLeungTerhal2002,
  author  = {DiVincenzo, D. P. and Leung, D. W. and Terhal, B. M.},
  title   = {{Quantum Data Hiding}},
  journal = {IEEE Trans. Inf. Theory},
  volume  = {48},
  pages   = {580--598},
  year    = {2002},
  doi     = {10.1109/18.985948},
  url     = {https://doi.org/10.1109/18.985948}
}

@article{Duan2009,
  author  = {Duan, R. and Feng, Y. and Xin, Y. and Ying, M.},
  title   = {{Distinguishability of Quantum States by Separable Operations}},
  journal = {IEEE Trans. Inf. Theory},
  volume  = {55},
  pages   = {1320--1330},
  year    = {2009},
  doi     = {10.1109/TIT.2008.2011524},
  url     = {https://doi.org/10.1109/TIT.2008.2011524}
}

@article{Edmonds1965,
  author = {Edmonds, J.},
  title = {Minimum Partition of a Matroid into Independent Subsets},
  journal = {J. Res. Natl. Bur. Stand. B}, volume = {69B}, pages = {67--72}, year = {1965},
  doi = {10.6028/jres.069B.004}, url = {https://doi.org/10.6028/jres.069B.004}
}

@article{EggelingWerner2002,
  author  = {Eggeling, T. and Werner, R. F.},
  title   = {{Hiding Classical Data in Multipartite Quantum States}},
  journal = {Phys. Rev. Lett.},
  volume  = {89},
  pages   = {097905},
  year    = {2002},
  doi     = {10.1103/PhysRevLett.89.097905},
  url     = {https://doi.org/10.1103/PhysRevLett.89.097905}
}

@article{Fan2004,
  author  = {Fan, H.},
  title   = {{Distinguishability and Indistinguishability by Local Operations and Classical Communication}},
  journal = {Phys. Rev. Lett.},
  volume  = {92},
  pages   = {177905},
  year    = {2004},
  doi     = {10.1103/PhysRevLett.92.177905},
  url     = {https://doi.org/10.1103/PhysRevLett.92.177905}
}

@article{Ghosh2001,
  author  = {Ghosh, S. and Kar, G. and Roy, A. and Sen(De), A. and Sen, U.},
  title   = {{Distinguishability of Bell States}},
  journal = {Phys. Rev. Lett.},
  volume  = {87},
  pages   = {277902},
  year    = {2001},
  doi     = {10.1103/PhysRevLett.87.277902},
  url     = {https://doi.org/10.1103/PhysRevLett.87.277902}
}

@article{Ghosh2004,
  author  = {Ghosh, S. and Kar, G. and Roy, A. and Sarkar, D.},
  title   = {{Distinguishability of Maximally Entangled States}},
  journal = {Phys. Rev. A},
  volume  = {70},
  pages   = {022304},
  year    = {2004},
  doi     = {10.1103/PhysRevA.70.022304},
  url     = {https://doi.org/10.1103/PhysRevA.70.022304}
}

@article{Gross2006,
  author = {Gross, D.},
  title = {Hudson's theorem for finite-dimensional quantum systems},
  journal = {J. Math. Phys.}, volume = {47}, pages = {122107}, year = {2006},
  doi = {10.1063/1.2393152}, url = {https://doi.org/10.1063/1.2393152}
}

@article{GrossNezamiWalter2021,
  author = {Gross, D. and Nezami, S. and Walter, M.},
  title = {{Schur–Weyl Duality for the Clifford Group with Applications: Property Testing, a Robust Hudson Theorem, and de Finetti Representations}},
  journal = {Commun. Math. Phys.}, volume = {385}, pages = {1325--1393}, year = {2021},
  doi = {10.1007/s00220-021-04118-7}, url = {https://doi.org/10.1007/s00220-021-04118-7}
}

@article{Halder2019,
  author  = {Halder, S. and Banik, M. and Agrawal, S. and Bandyopadhyay, S.},
  title   = {{Strong Quantum Nonlocality without Entanglement}},
  journal = {Phys. Rev. Lett.},
  volume  = {122},
  pages   = {040403},
  year    = {2019},
  doi     = {10.1103/PhysRevLett.122.040403},
  url     = {https://doi.org/10.1103/PhysRevLett.122.040403}
}

@article{Hashimoto2021,
  author  = {Hashimoto, T. and Horibe, M. and Hayashi, A.},
  title   = {{Simple Criterion for Local Distinguishability of Generalized Bell States in Prime Dimension}},
  journal = {Phys. Rev. A},
  volume  = {103},
  pages   = {052429},
  year    = {2021},
  doi     = {10.1103/PhysRevA.103.052429},
  url     = {https://doi.org/10.1103/PhysRevA.103.052429}
}

@article{Hayashi2006,
  author  = {Hayashi, M. and Markham, D. and Murao, M. and Owari, M. and Virmani, S.},
  title   = {{Bounds on Multipartite Entangled Orthogonal State Discrimination Using Local Operations and Classical Communication}},
  journal = {Phys. Rev. Lett.},
  volume  = {96},
  pages   = {040501},
  year    = {2006},
  doi     = {10.1103/PhysRevLett.96.040501},
  url     = {https://doi.org/10.1103/PhysRevLett.96.040501}
}

@article{Helwig2012,
  author  = {Helwig, W. and Cui, W. and Latorre, J. I. and Riera, A. and Lo, H.-K.},
  title   = {Absolute Maximal Entanglement and Quantum Secret Sharing},
  journal = {Phys. Rev. A},
  volume  = {86},
  pages   = {052335},
  year    = {2012},
  doi     = {10.1103/PhysRevA.86.052335},
  url     = {https://doi.org/10.1103/PhysRevA.86.052335}
}

@article{HeShiZhang2024,
  author = {He, Y. and Shi, F. and Zhang, X.},
  title = {Strong quantum nonlocality and unextendibility without entanglement in $N$-partite systems with odd $N$},
  journal = {Quantum}, volume = {8}, pages = {1349}, year = {2024},
  doi = {10.22331/q-2024-05-16-1349}, url = {https://doi.org/10.22331/q-2024-05-16-1349}
}

@article{Higgins2011,
  author  = {Higgins, B. L. and Doherty, A. C. and Bartlett, S. D. and Pryde, G. J. and Wiseman, H. M.},
  title   = {{Multiple-Copy State Discrimination: Thinking Globally, Acting Locally}},
  journal = {Phys. Rev. A},
  volume  = {83},
  pages   = {052314},
  year    = {2011},
  doi     = {10.1103/PhysRevA.83.052314},
  url     = {https://doi.org/10.1103/PhysRevA.83.052314}
}

@article{Horodecki2003,
  author  = {Horodecki, M. and Sen(De), A. and Sen, U. and Horodecki, K.},
  title   = {{Local Indistinguishability: More Nonlocality with Less Entanglement}},
  journal = {Phys. Rev. Lett.},
  volume  = {90},
  pages   = {047902},
  year    = {2003},
  doi     = {10.1103/PhysRevLett.90.047902},
  url     = {https://doi.org/10.1103/PhysRevLett.90.047902}
}

@article{HuGaoYan2024,
  author  = {Hu, M. and Gao, T. and Yan, F.},
  title   = {{Strong Quantum Nonlocality with Genuine Entanglement in an $N$-Qutrit System}},
  journal = {Phys. Rev. A},
  volume  = {109},
  pages   = {022220},
  year    = {2024},
  doi     = {10.1103/PhysRevA.109.022220},
  url     = {https://doi.org/10.1103/PhysRevA.109.022220}
}

@article{HuGaoYan2025,
  author  = {Hu, M. and Gao, T. and Yan, F.},
  title   = {{Strongest Quantum Nonlocality in $N$-Partite Systems}},
  journal = {Phys. Rev. A},
  volume  = {112},
  pages   = {032205},
  year    = {2025},
  doi     = {10.1103/69z8-sk86},
  url     = {https://doi.org/10.1103/69z8-sk86}
}

@article{Jagannathan2022,
  author  = {Jagannathan, A. and Grace, M. and Brasher, O. and Shapiro, J. H. and Guha, S. and Habif, J. L.},
  title   = {{Demonstration of Quantum-Limited Discrimination of Multicopy Pure versus Mixed States}},
  journal = {Phys. Rev. A},
  volume  = {105},
  pages   = {032446},
  year    = {2022},
  doi     = {10.1103/PhysRevA.105.032446},
  url     = {https://doi.org/10.1103/PhysRevA.105.032446}
}

@misc{KuengGross2015,
  author = {Kueng, R. and Gross, D.},
  title = {Qubit stabilizer states are complex projective 3-designs},
  year = {2015}, eprint = {1510.02767}, archivePrefix = {arXiv}, primaryClass = {quant-ph},
  url = {https://arxiv.org/abs/1510.02767}
}

@article{LiShiWang2022,
  author  = {Li, M.-S. and Shi, F. and Wang, Y.-L.},
  title   = {{Local Discrimination of Generalized Bell States via Commutativity}},
  journal = {Phys. Rev. A},
  volume  = {105},
  pages   = {032455},
  year    = {2022},
  doi     = {10.1103/PhysRevA.105.032455},
  url     = {https://doi.org/10.1103/PhysRevA.105.032455}
}

@article{LiShiZhang2023,
  author = {Li, J. and Shi, F. and Zhang, X.},
  title = {Strongest nonlocal sets with small sizes},
  journal = {Phys. Rev. A}, volume = {108}, pages = {062407}, year = {2023},
  doi = {10.1103/PhysRevA.108.062407}, url = {https://doi.org/10.1103/PhysRevA.108.062407}
}

@article{LiZheng2022,
  author  = {Li, M.-S. and Zheng, Z.-J.},
  title   = {{Genuine Hidden Nonlocality without Entanglement: From the Perspective of Local Discrimination}},
  journal = {New J. Phys.},
  volume  = {24},
  pages   = {043036},
  year    = {2022},
  doi     = {10.1088/1367-2630/ac631a},
  url     = {https://doi.org/10.1088/1367-2630/ac631a}
}

@article{LuCaoZuoFei2024,
  author = {Lu, Y.-Y. and Cao, H.-Q. and Zuo, H.-J. and Fei, S.-M.},
  title = {Genuinely nonlocal sets without entanglement in multipartite systems},
  journal = {Phys. Rev. A}, volume = {110}, pages = {022427}, year = {2024},
  doi = {10.1103/PhysRevA.110.022427}, url = {https://doi.org/10.1103/PhysRevA.110.022427}
}

@article{Matthews2009,
  author  = {Matthews, W. and Wehner, S. and Winter, A.},
  title   = {{Distinguishability of Quantum States under Restricted Families of Measurements with an Application to Quantum Data Hiding}},
  journal = {Commun. Math. Phys.},
  volume  = {291},
  pages   = {813--843},
  year    = {2009},
  doi     = {10.1007/s00220-009-0890-5},
  url     = {https://doi.org/10.1007/s00220-009-0890-5}
}

@article{Mele2024Haar,
  author  = {Mele, A. A.},
  title   = {{Introduction to Haar Measure Tools in Quantum Information: A Beginner's Tutorial}},
  journal = {Quantum},
  volume  = {8},
  pages   = {1340},
  year    = {2024},
  doi     = {10.22331/q-2024-05-08-1340},
  url     = {https://doi.org/10.22331/q-2024-05-08-1340}
}

@article{MurshidEtAl2026,
  author  = {Murshid, S. and Gupta, T. and Russo, V. and Bandyopadhyay, S.},
  title   = {{Quantum Nonlocality without Entanglement and State Discrimination Measures}},
  journal = {Quantum},
  volume  = {10},
  pages   = {2174},
  year    = {2026},
  doi     = {10.22331/q-2026-07-23-2174},
  url     = {https://doi.org/10.22331/q-2026-07-23-2174}
}

@article{PeresWootters1991,
  author  = {Peres, A. and Wootters, W. K.},
  title   = {{Optimal Detection of Quantum Information}},
  journal = {Phys. Rev. Lett.},
  volume  = {66},
  pages   = {1119--1122},
  year    = {1991},
  doi     = {10.1103/PhysRevLett.66.1119},
  url     = {https://doi.org/10.1103/PhysRevLett.66.1119}
}

@article{Rout2019,
  author  = {Rout, S. and Maity, A. G. and Mukherjee, A. and Halder, S. and Banik, M.},
  title   = {{Genuinely Nonlocal Product Bases: Classification and Entanglement-Assisted Discrimination}},
  journal = {Phys. Rev. A},
  volume  = {100},
  pages   = {032321},
  year    = {2019},
  doi     = {10.1103/PhysRevA.100.032321},
  url     = {https://doi.org/10.1103/PhysRevA.100.032321}
}

@article{RoutEtAl2021,
  author  = {Rout, S. and Maity, A. G. and Mukherjee, A. and Halder, S. and Banik, M.},
  title   = {{Multiparty Orthogonal Product States with Minimal Genuine Nonlocality}},
  journal = {Phys. Rev. A},
  volume  = {104},
  pages   = {052433},
  year    = {2021},
  doi     = {10.1103/PhysRevA.104.052433},
  url     = {https://doi.org/10.1103/PhysRevA.104.052433}
}

@article{SenHalderSen2024,
  author  = {Sen, K. and Halder, S. and Sen, U.},
  title   = {{Incompatibility of Local Measurements Providing an Advantage in Local Quantum State Discrimination}},
  journal = {Phys. Rev. A},
  volume  = {109},
  pages   = {012415},
  year    = {2024},
  doi     = {10.1103/PhysRevA.109.012415},
  url     = {https://doi.org/10.1103/PhysRevA.109.012415}
}

@article{ShiStrongUPB2022,
  author  = {Shi, F. and Li, M.-S. and Hu, M. and Chen, L. and Yung, M.-H. and Wang, Y.-L. and Zhang, X.},
  title   = {{Strongly Nonlocal Unextendible Product Bases Do Exist}},
  journal = {Quantum},
  volume  = {6},
  pages   = {619},
  year    = {2022},
  doi     = {10.22331/q-2022-01-05-619},
  url     = {https://doi.org/10.22331/q-2022-01-05-619}
}

@article{Terhal2001,
  author  = {Terhal, B. M. and DiVincenzo, D. P. and Leung, D. W.},
  title   = {{Hiding Bits in Bell States}},
  journal = {Phys. Rev. Lett.},
  volume  = {86},
  pages   = {5807--5810},
  year    = {2001},
  doi     = {10.1103/PhysRevLett.86.5807},
  url     = {https://doi.org/10.1103/PhysRevLett.86.5807}
}

@article{TianEtAl2024,
  author  = {Tian, B. and Yan, W.-Z. and Hou, Z. and Xiang, G.-Y. and Li, C.-F. and Guo, G.-C.},
  title   = {{Minimum-Consumption Discrimination of Quantum States via Globally Optimal Adaptive Measurements}},
  journal = {Phys. Rev. Lett.},
  volume  = {132},
  pages   = {110801},
  year    = {2024},
  doi     = {10.1103/PhysRevLett.132.110801},
  url     = {https://doi.org/10.1103/PhysRevLett.132.110801}
}

@article{Walgate2000,
  author  = {Walgate, J. and Short, A. J. and Hardy, L. and Vedral, V.},
  title   = {{Local Distinguishability of Multipartite Orthogonal Quantum States}},
  journal = {Phys. Rev. Lett.},
  volume  = {85},
  pages   = {4972--4975},
  year    = {2000},
  doi     = {10.1103/PhysRevLett.85.4972},
  url     = {https://doi.org/10.1103/PhysRevLett.85.4972}
}

@article{WalgateHardy2002,
  author  = {Walgate, J. and Hardy, L.},
  title   = {{Nonlocality, Asymmetry, and Distinguishing Bipartite States}},
  journal = {Phys. Rev. Lett.},
  volume  = {89},
  pages   = {147901},
  year    = {2002},
  doi     = {10.1103/PhysRevLett.89.147901},
  url     = {https://doi.org/10.1103/PhysRevLett.89.147901}
}

@article{WangEtAl2025,
  author  = {Wang, C.-H. and Yuan, J.-T. and Yang, Y.-H. and Li, M.-S. and Fei, S.-M. and Ma, Z.-H.},
  title   = {{Detectors for Local Discrimination of Sets of Generalized Bell States}},
  journal = {Phys. Rev. A},
  volume  = {111},
  pages   = {042408},
  year    = {2025},
  doi     = {10.1103/PhysRevA.111.042408},
  url     = {https://doi.org/10.1103/PhysRevA.111.042408}
}

@article{Watrous2005,
  author  = {Watrous, J.},
  title   = {{Bipartite Subspaces Having No Bases Distinguishable by Local Operations and Classical Communication}},
  journal = {Phys. Rev. Lett.},
  volume  = {95},
  pages   = {080505},
  year    = {2005},
  doi     = {10.1103/PhysRevLett.95.080505},
  url     = {https://doi.org/10.1103/PhysRevLett.95.080505}
}

@article{WalgateScott2008,
  author  = {Walgate, Jonathan and Scott, A. J.},
  title   = {Generic local distinguishability and completely entangled subspaces},
  journal = {Physical Review A},
  volume  = {77},
  number  = {6},
  pages   = {062317},
  year    = {2008},
  doi     = {10.1103/PhysRevA.77.062317}
}

@article{Webb2016,
  author = {Webb, Z.},
  title = {The Clifford group forms a unitary 3-design},
  journal = {Quantum Inf. Comput.}, volume = {16}, pages = {1379--1400}, year = {2016},
  url = {https://www.rintonpress.com/xxqic16/qic-16-1516/1379-1400.pdf}
}

@article{XiongLiZhengLi2023,
  author  = {Xiong, Z.-X. and Li, M.-S. and Zheng, Z.-J. and Li, L.},
  title   = {{Distinguishability-Based Genuine Nonlocality with Genuine Multipartite Entanglement}},
  journal = {Phys. Rev. A},
  volume  = {108},
  pages   = {022405},
  year    = {2023},
  doi     = {10.1103/PhysRevA.108.022405},
  url     = {https://doi.org/10.1103/PhysRevA.108.022405}
}

@article{Yu2012,
  author  = {Yu, N. and Duan, R. and Ying, M.},
  title   = {{Four Locally Indistinguishable Ququad-Ququad Orthogonal Maximally Entangled States}},
  journal = {Phys. Rev. Lett.},
  volume  = {109},
  pages   = {020506},
  year    = {2012},
  doi     = {10.1103/PhysRevLett.109.020506},
  url     = {https://doi.org/10.1103/PhysRevLett.109.020506}
}

@article{ZhangStrong2019,
  author = {Zhang, Z.-C. and Zhang, X.},
  title = {Strong quantum nonlocality in multipartite quantum systems},
  journal = {Phys. Rev. A}, volume = {99}, pages = {062108}, year = {2019},
  doi = {10.1103/PhysRevA.99.062108}, url = {https://doi.org/10.1103/PhysRevA.99.062108}
}

@article{ZhenFeiZuo2022,
  author = {Zhen, X.-F. and Fei, S.-M. and Zuo, H.-J.},
  title = {Nonlocality without entanglement in general multipartite quantum systems},
  journal = {Phys. Rev. A}, volume = {106}, pages = {062432}, year = {2022},
  doi = {10.1103/PhysRevA.106.062432}, url = {https://doi.org/10.1103/PhysRevA.106.062432}
}

@article{Zhou2023,
  author  = {Zhou, H. and Gao, T. and Yan, F.},
  title   = {{Strong Quantum Nonlocality without Entanglement in an $n$-Partite System with Even $n$}},
  journal = {Phys. Rev. A},
  volume  = {107},
  pages   = {042214},
  year    = {2023},
  doi     = {10.1103/PhysRevA.107.042214},
  url     = {https://doi.org/10.1103/PhysRevA.107.042214}
}

@article{ZhouPlane2022,
  author  = {Zhou, H. and Gao, T. and Yan, F.},
  title   = {{Orthogonal Product Sets with Strong Quantum Nonlocality on a Plane Structure}},
  journal = {Phys. Rev. A},
  volume  = {106},
  pages   = {052209},
  year    = {2022},
  doi     = {10.1103/PhysRevA.106.052209},
  url     = {https://doi.org/10.1103/PhysRevA.106.052209}
}

@article{Zhu2017,
  author = {Zhu, H.},
  title = {{Multiqubit Clifford groups are unitary 3-designs}},
  journal = {Phys. Rev. A}, volume = {96}, pages = {062336}, year = {2017},
  doi = {10.1103/PhysRevA.96.062336}, url = {https://doi.org/10.1103/PhysRevA.96.062336}
}

\clearpage
\onecolumngrid
\begin{center}
{\large\bfseries Supplemental Material for ``Power and Limits of Collective Local Measurements in Multicopy State Discrimination''}
\end{center}
\vspace{0.5em}

\setcounter{section}{0}
\renewcommand{\thesection}{\Alph{section}}
\renewcommand{\theequation}{\thesection\arabic{equation}}
\renewcommand{\thelemma}{\thesection\arabic{lemma}}
\renewcommand{\theproposition}{\thesection\arabic{proposition}}
\renewcommand{\thetheorem}{\thesection\arabic{theorem}}
\renewcommand{\thecorollary}{\thesection\arabic{corollary}}
\providecommand{\theHequation}{}
\providecommand{\theHlemma}{}
\providecommand{\theHproposition}{}
\providecommand{\theHtheorem}{}
\providecommand{\theHcorollary}{}
\renewcommand{\theHequation}{supp.\thesection.\arabic{equation}}
\renewcommand{\theHlemma}{supp.\thesection.\arabic{lemma}}
\renewcommand{\theHproposition}{supp.\thesection.\arabic{proposition}}
\renewcommand{\theHtheorem}{supp.\thesection.\arabic{theorem}}
\renewcommand{\theHcorollary}{supp.\thesection.\arabic{corollary}}
\newcommand{\suppsection}[2]{%
  \refstepcounter{section}%
  \setcounter{equation}{0}%
  \setcounter{lemma}{0}%
  \setcounter{proposition}{0}%
  \setcounter{theorem}{0}%
  \setcounter{corollary}{0}%
  \section*{SM-\thesection:\ #1}\label{#2}%
}

The Supplemental Material is organized as follows. SM-A gives the precise measurement hierarchy and proves the trace--overlap bounds. SM-B treats the $\operatorname{AME}(6,2)$ local-Pauli orbit basis.  SM-C proves the moment-to-overlap uniformization lemma. SM-D introduces the odd-prime Weyl--Pauli stabilizer formalism, derives the stabilizer moment estimate, and proves Theorem~\ref{thm:individual-unbounded}. SM-E proves the collective stabilizer decoder of Theorem~\ref{thm:stabilizer-collective}. SM-F proves the accessible-information inequality. SM-G develops the flat local-unitary orbit and the Newton-identity replica-rigidity lemma. SM-H derives the moment-curve replica bound and proves Theorem~\ref{thm:collective-unbounded}.

Unless another limit is stated explicitly, all asymptotic statements refer to $n\to\infty$ with the physical local dimensions fixed. For a real sequence $f_n$ and a positive sequence $g_n$, $f_n=O(g_n)$ means $|f_n|\le Cg_n$ for all sufficiently large $n$ and some constant $C$ independent of $n$, while $f_n=o(g_n)$ means $f_n/g_n\to0$. A subscript on $O$, such as $O_p$ or $O_d$ below, indicates that the implicit constant may depend only on the displayed parameter. We write $f_n\asymp g_n$ for positive sequences when both $f_n=O(g_n)$ and $g_n=O(f_n)$. Thus $O(1)$ and $O_d(1)$ denote bounded quantities (with the indicated parameter dependence), whereas $o(1)$ denotes a quantity tending to zero. The same conventions apply when another variable is explicitly taken to infinity. We use $\ln$ for the natural logarithm and $\log_b$ for logarithm to base $b$.

\suppsection{Measurement hierarchy and trace--overlap bounds}{app:trace}

We first make the measurement models precise and then prove Proposition~\ref{prop:trace-overlap}.  The arguments use only finite-dimensional separability and the normalization of a POVM.

For the discrimination probabilities $P_{\alpha}^{\mathsf M}$, only the $D$-outcome decision POVMs $\{M_y\}_{y=1}^D$ of the Letter are needed.  In SM-F we also allow a measurement to have a general classical outcome in a standard Borel space $(\mathcal Z,\Sigma)$.  Such a POVM is a countably additive map $M:\Sigma\to\operatorname{Pos}(\cH^{\otimes r})$ with $M(\mathcal Z)=I$, where $\operatorname{Pos}(\mathcal V)$ denotes the cone of positive semidefinite operators on a Hilbert space $\mathcal V$.  For SEP, we require $M(E)$ to belong to the corresponding separable cone for every measurable $E\in\Sigma$; LOCC measurements are, as usual, closed under measurable classical relabeling of their terminal outcomes.  Hence if $E_1,\ldots,E_D$ is a measurable partition of $\mathcal Z$, the coarse-grained decision POVM $\{M(E_y)\}_{y=1}^D$ remains in the same measurement class.  This is the only general-outcome property used in SM-F.

\subsection{Individual-copy and collective measurement classes}
For one copy, let $\SEP_+(\cH)$ denote the cone of positive operators separable across the physical partition $\cH=\bigotimes_{j=1}^n\cH_j$.  For $r$ copies, define the individual-copy separable cone by
\begin{equation}
 \mathsf C_{\ind}^{(r)}
 :=\operatorname{cone}\!\left\{
 A_1\otimes\cdots\otimes A_r:
 A_a\in\SEP_+(\cH)
 \right\}.
 \label{eq:A-indcone}
\end{equation}
Here $\operatorname{cone}$ denotes the conic hull, i.e., the set of all finite nonnegative linear combinations of the displayed generators.  An $r$-copy POVM belongs to $\SEP_{\ind}^{(r)}$ when each POVM effect lies in $\mathsf C_{\ind}^{(r)}$.  By contrast, $\SEP_{\col}^{(r)}$ is ordinary separability with respect to the grouped partition
\begin{equation}
 \cH^{\otimes r}\simeq\bigotimes_{j=1}^n \cH_j^{\otimes r}.
 \label{eq:A-grouped}
\end{equation}
The difference is important: an element of $\SEP_{\col}^{(r)}$ may be entangled among the $r$ local copies held by one laboratory, whereas an element of $\mathsf C_{\ind}^{(r)}$ is built from one-copy separable factors before the copies are combined.

Consider now a terminal branch $\tau$ of a finite-round protocol in $\LOCC_{\ind}^{(r)}$.  Every elementary Kraus operator on that branch acts on one laboratory and one copy.  Operators acting on distinct tensor factors commute, so all operations belonging to the same copy can be regrouped.  The branch Kraus operator therefore has the form
\begin{equation}
 K_\tau=K_{\tau,1}\otimes\cdots\otimes K_{\tau,r}.
 \label{eq:branch-regroup}
\end{equation}
For each copy $a$, $K_{\tau,a}$ is itself a product of local Kraus operators across the $n$ laboratories.  Consequently
\begin{equation}
 A_{\tau,a}:=K_{\tau,a}^\dagger K_{\tau,a}\in\SEP_+(\cH),
 \label{eq:A-branch-onecopy}
\end{equation}
and the terminal effect satisfies
\begin{equation}
 E_\tau=K_\tau^\dagger K_\tau
 =A_{\tau,1}\otimes\cdots\otimes A_{\tau,r}
 \in\mathsf C_{\ind}^{(r)}.
 \label{eq:A-terminal-effect}
\end{equation}
A decision outcome is obtained by summing terminal effects over all branches that lead to the same decision.  Since $\mathsf C_{\ind}^{(r)}$ is a cone, the coarse-grained effect remains in $\mathsf C_{\ind}^{(r)}$.  Hence
\begin{equation}
 \LOCC_{\ind}^{(r)}\subseteq\SEP_{\ind}^{(r)}.
 \label{eq:A-locc-sep-ind}
\end{equation}
Adaptive LOCC imposes the additional requirement that copy $a$ be completed before copy $a+1$ is touched, which gives $\LOCC_{\ad}^{(r)}\subseteq\LOCC_{\ind}^{(r)}$.  Finally, after regrouping the tensor factors by laboratory, every generator $A_1\otimes\cdots\otimes A_r$ of $\mathsf C_{\ind}^{(r)}$ is separable across the grouped partition.  Thus $\SEP_{\ind}^{(r)}\subseteq\SEP_{\col}^{(r)}$.  The inclusion $\LOCC_{\ind}^{(r)}\subseteq\LOCC_{\col}^{(r)}$ is immediate because the collective model allows every local operation available in the individual-copy model.  These observations establish Eq.~\eqref{eq:model-inclusions}.

The optimal success probability is monotone in the number of available copies.

\begin{lemma}[Monotonicity in copy number]\label{lem:A-copy-monotonicity}
For every complete basis $\cB$, every architecture $\alpha\in\{\ind,\col\}$, every spatial class $\mathsf M\in\{\LOCC,\SEP\}$, and integers $1\le r\le s$,
\begin{equation}
 P_{\alpha}^{\mathsf M}(r;\cB)\le P_{\alpha}^{\mathsf M}(s;\cB).
 \label{eq:A-copy-monotonicity}
\end{equation}
Consequently, $P_{\alpha}^{\mathsf M}(R;\cB)<1$ implies $N_{\alpha}^{\mathsf M}(\cB)>R$.
\end{lemma}

\begin{proof}
An $r$-copy protocol can be applied when $s\ge r$ copies are available by discarding $s-r$ copies locally before the protocol starts.  Local discarding is an allowed operation in each of the four measurement models, so every $r$-copy decision rule is also realizable with $s$ copies and has the same success probability.  Taking the supremum over the $s$-copy class proves Eq.~\eqref{eq:A-copy-monotonicity}.  The final statement follows immediately from the definition of copy complexity.
\end{proof}

\subsection{Collective SEP bound}
Let $\{M_y\}_{y=1}^{D}$ be a $\SEP_{\col}^{(r)}$ POVM.  In finite dimension, a positive separable operator is a finite sum of positive product operators.  Spectrally decomposing every local positive factor refines such a decomposition into rank-one product projectors, so we may write
\begin{equation}
 M_y=\sum_a\lambda_{y,a}\proj{\Phi_{y,a}},
 \qquad \lambda_{y,a}\ge0,
 \label{eq:A-collective-decomp}
\end{equation}
where
\begin{equation}
 \ket{\Phi_{y,a}}=\bigotimes_{j=1}^n\ket{\Phi_{y,a}^{(j)}},
 \qquad \ket{\Phi_{y,a}^{(j)}}\in\cH_j^{\otimes r}.
\end{equation}
The product vectors may be taken normalized, with their norms absorbed into $\lambda_{y,a}$.  By the definition of $\Gamma_r(\cB)$,
\begin{equation}
 \Tr(M_y\rho_y^{\otimes r})=\sum_a\lambda_{y,a}|\langle\Phi_{y,a}|\psi_y^{\otimes r}\rangle|^2\le\Gamma_r(\cB)\sum_a\lambda_{y,a}.
 \label{eq:A-collective-oneeffect}
\end{equation}
Because every rank-one projector in Eq.~\eqref{eq:A-collective-decomp} has unit trace,
$\sum_a\lambda_{y,a}=\Tr M_y$.  Therefore
\begin{equation}
 \Tr(M_y\rho_y^{\otimes r})\le \Gamma_r(\cB)\Tr M_y.
 \label{eq:A-collective-trace}
\end{equation}
Averaging over the uniform label and using $\sum_yM_y=I_{D^r}$ gives
\begin{equation}
 P_{\col}^{\SEP}(r;\cB)\le\frac{\Gamma_r(\cB)}{D}\sum_y\Tr M_y=\frac{\Gamma_r(\cB)}{D}\Tr I_{D^r}=D^{r-1}\Gamma_r(\cB).
 \label{eq:A-collective-final}
\end{equation}
which is Eq.~\eqref{eq:collective-trace-main}.

\subsection{Individual-copy SEP bound}
We first record the corresponding one-copy estimate.  Let $A\in\SEP_+(\cH)$.  As above, $A$ has a rank-one fully product decomposition
\begin{equation}
 A=\sum_s\lambda_s\proj{\phi_s},
 \qquad \lambda_s\ge0,
 \label{eq:A-onecopy-decomp}
\end{equation}
with $\ket{\phi_s}=\bigotimes_j\ket{\phi_{s,j}}$.  Hence
\begin{equation}
 \Tr(A\rho_y)=\sum_s\lambda_s|\langle\phi_s|\psi_y\rangle|^2\le\gamma(\cB)\sum_s\lambda_s=\gamma(\cB)\Tr A.
 \label{eq:onecopy-sep-gamma}
\end{equation}
Now let $M_y$ be an individual-copy SEP effect.  From Eq.~\eqref{eq:A-indcone},
\begin{equation}
 M_y=\sum_a A_{a,1}\otimes\cdots\otimes A_{a,r},
 \qquad A_{a,\ell}\in\SEP_+(\cH).
 \label{eq:A-ind-effect}
\end{equation}
Since the unknown state is the tensor power $\rho_y^{\otimes r}$, the trace factorizes copy by copy:
\begin{equation}
 \Tr(M_y\rho_y^{\otimes r})=\sum_a\prod_{\ell=1}^r\Tr(A_{a,\ell}\rho_y)\le\gamma(\cB)^r\sum_a\prod_{\ell=1}^r\Tr A_{a,\ell}=\gamma(\cB)^r\Tr M_y.
 \label{eq:A-ind-trace}
\end{equation}
Summing over the POVM outcomes gives
\begin{equation}
 P_{\ind}^{\SEP}(r;\cB)
 \le\frac{\gamma(\cB)^r}{D}\Tr I_{D^r}
 =D^{r-1}\gamma(\cB)^r,
 \label{eq:A-ind-final}
\end{equation}
which is Eq.~\eqref{eq:copytrace-main} and completes the proof of Proposition~\ref{prop:trace-overlap}.

\subsection{Local-unitary invariance and the overlap floor}
Let $U=\bigotimes_jU_j$ be local.  The map
$\ket{\phi}=\bigotimes_j\ket{\phi_j}\mapsto U^\dagger\ket{\phi}$
is a bijection of the normalized product states.  Therefore
\begin{equation}
 \gamma(U\ket{\psi})=\max_{\ket{\phi}\,\mathrm{product}}|\langle\phi|U|\psi\rangle|^2=\max_{\ket{\phi}\,\mathrm{product}}|\langle U^\dagger\phi|\psi\rangle|^2=\gamma(\ket{\psi}).
 \label{eq:A-gamma-LU}
\end{equation}
On $r$ copies,
$U^{\otimes r}=\bigotimes_jU_j^{\otimes r}$ is local with respect to the grouped partition, and the same argument gives
\begin{equation}
 \Gamma_r(U\ket{\psi})=\Gamma_r(\ket{\psi}).
 \label{eq:A-Gamma-LU}
\end{equation}
This proves Eq.~\eqref{eq:LU-overlap-invariance}.

Finally, fix any orthonormal product basis $\{\ket{e_x}\}_{x=1}^D$ of $\cH$.  Normalization gives
\begin{equation}
 \sum_{x=1}^{D}|\langle e_x|\psi\rangle|^2=1.
\end{equation}
At least one term is therefore at least $1/D$, and since the maximization defining $\gamma$ includes every vector in this product basis,
\begin{equation}
 \gamma(\ket{\psi})\ge D^{-1}.
 \label{eq:A-gamma-floor}
\end{equation}

\suppsection{\texorpdfstring{$\operatorname{AME}(6,2)$}{AME(6,2)} local-Pauli orbit bases}{app:ame}

We prove the two statements used in Eqs.~\eqref{eq:AME-orbit-basis}--\eqref{eq:AME-gamma}.  Recall that a pure state $\ket{\Psi}\in(\C^d)^{\otimes N}$ is absolutely maximally entangled, or $\operatorname{AME}(N,d)$, when every reduction to at most $\lfloor N/2\rfloor$ parties is maximally mixed.  Equivalently, the state is maximally entangled across every bipartition, with maximal Schmidt rank on the smaller side~\cite{Helwig2012}.  For $\operatorname{AME}(6,2)$, every three-qubit marginal is therefore $I_8/8$.

\subsection{Every \texorpdfstring{$\operatorname{AME}(6,2)$}{AME(6,2)} state generates a complete local-Pauli orbit basis}

Fix any three-party subset $A\subset\{1,\ldots,6\}$ and let $\bar A$ be its complement.  Here $\F_2=\{0,1\}$ with addition and multiplication modulo two. On one qubit define the Pauli shift and phase operators $X_2\ket q=\ket{q+1\bmod2}$ and $Z_2\ket q=(-1)^q\ket q$, and let $X_{2,j},Z_{2,j}$ denote their actions on party $j$. Define
\begin{equation}
 \mathcal P_A:=\left\{\bigotimes_{j\in A}X_{2,j}^{a_j}Z_{2,j}^{b_j}:a_j,b_j\in\F_2\right\},\qquad |\mathcal P_A|=4^3=64.
 \label{eq:C-Pauli-set}
\end{equation}
Also write $\rho_A:=\Tr_{\bar A}\proj{\Psi}=I_8/8$.  For $P_A,Q_A\in\mathcal P_A$, the AME property gives
\begin{equation}
 \langle\Psi|(P_A^\dagger Q_A\otimes I_{\bar A})|\Psi\rangle=\Tr(\rho_A P_A^\dagger Q_A)=\frac18\Tr(P_A^\dagger Q_A)=\delta_{P_A,Q_A}.
 \label{eq:C-AME-orbit-orthogonality}
\end{equation}
Here we used the Hilbert--Schmidt orthogonality of the three-qubit Pauli error basis.  Thus the $64$ states $(P_A\otimes I_{\bar A})\ket{\Psi}$ are mutually orthonormal.  Since the six-qubit Hilbert space has dimension $2^6=64$, they form a complete orthonormal basis.  Every orbit element differs from $\ket{\Psi}$ by a tensor product of one-qubit unitaries, so Eq.~\eqref{eq:LU-overlap-invariance} gives the same $\gamma$ and $\Gamma_r$ for all basis states.

\subsection{Strict product-overlap bound}

\begin{lemma}[Strict AME product-overlap bound]\label{lem:AME-product-strict}
If $\ket{\Psi}$ is $\operatorname{AME}(6,2)$, then
\begin{equation}
 \gamma(\ket{\Psi})
 =\max_{\substack{\ket{\phi_1},\ldots,\ket{\phi_6}\\ \langle\phi_j|\phi_j\rangle=1}}
 \left|\left\langle\bigotimes_{j=1}^{6}\phi_j\middle|\Psi\right\rangle\right|^2
 <\frac18.
 \label{eq:C-AME-strict}
\end{equation}
\end{lemma}

\begin{proof}
The normalized fully product states form a compact set and the overlap is continuous, so the maximum is attained.  Across any $3|3$ cut $S|\bar S$, the AME property gives eight equal Schmidt coefficients $1/\sqrt8$.  Every fully product state is product across this cut, hence its squared overlap with $\ket{\Psi}$ is at most $1/8$.

Suppose equality is attained by a fully product state.  Apply local one-qubit unitaries mapping that product state to $\ket{0}^{\otimes6}$.  The AME property is invariant under local unitaries, so, relabeling the transformed state as $\ket{\Psi}$, write
\begin{equation}
 \ket{\Psi}=\sum_{\mathbf{x}\in\F_2^6}c_{\mathbf{x}}\ket{\mathbf{x}},\qquad |c_{\mathbf{0}}|^2=\frac18.
 \label{eq:C-AME-zero}
\end{equation}
Here $\mathbf{x}=(x_1,\ldots,x_6)$, $\ket{\mathbf{x}}:=\bigotimes_{j=1}^6\ket{x_j}$, and $\mathbf{0}:=(0,\ldots,0)$. For later use, $\operatorname{supp}(\mathbf{x})$ and $\operatorname{wt}(\mathbf{x})$ denote the support and Hamming weight of the bit string $\mathbf{x}$.
Fix a three-element subset $S$.  Writing $\mathbf{x}=(\mathbf{a},\mathbf{b})$ according to the cut $S|\bar S$, define the $8\times8$ coefficient matrix
\begin{equation}
 U_S:=\sqrt8\,[c_{\mathbf{a},\mathbf{b}}]_{\mathbf{a}\in\F_2^S,\,\mathbf{b}\in\F_2^{\bar S}}.
 \label{eq:C-AME-unitary}
\end{equation}
Since $\rho_S=I_8/8$, one has $U_SU_S^\dagger=I_8$, so $U_S$ is unitary.  Its $(0,0)$ entry has modulus one by Eq.~\eqref{eq:C-AME-zero}.  In a unitary matrix, an entry of modulus one forces every other entry in the same row and column to vanish.

Now let $\mathbf{x}\ne\mathbf{0}$ have Hamming weight at most three.  Choose a three-element set $S$ containing $\operatorname{supp}(\mathbf{x})$.  Then $\mathbf{x}=(\mathbf{a},\mathbf{0})$ with $\mathbf{a}\ne\mathbf{0}$, and $c_{\mathbf{x}}$ lies in the zero column of $U_S$ away from its $(\mathbf{0},\mathbf{0})$ entry.  Hence
\begin{equation}
 c_{\mathbf{x}}=0,\qquad 1\le \operatorname{wt}(\mathbf{x})\le3.
 \label{eq:C-AME-lowweight-zero}
\end{equation}
Finally use the fixed cut $S_0=\{1,2,3\}$ and $\bar S_0=\{4,5,6\}$.  In each of the three rows $\mathbf{a}=(1,0,0),(0,1,0),(0,0,1)$, Eq.~\eqref{eq:C-AME-lowweight-zero} forces every entry to vanish except possibly the column $\mathbf{b}=(1,1,1)$: if $\mathbf{b}\ne(1,1,1)$, then $\operatorname{wt}(\mathbf{a},\mathbf{b})\le3$.  But every row of the unitary $U_{S_0}$ has norm one, so each of these three rows must have a nonzero entry in that same single column.  Such rows cannot be mutually orthogonal, contradicting unitarity.  Therefore equality with $1/8$ is impossible, and Eq.~\eqref{eq:C-AME-strict} follows.
\end{proof}

Because every basis state in $\cB_A(\Psi)$ is a local-Pauli translate of $\ket{\Psi}$, local-unitary invariance yields
\begin{equation}
 \gamma(\cB_A(\Psi))=\gamma(\ket{\Psi})<\frac18.
 \label{eq:C-AME-basis-gamma}
\end{equation}
With $D=64$, Proposition~\ref{prop:trace-overlap} therefore gives
\begin{equation}
 P_{\ind}^{\SEP}(2;\cB_A(\Psi))
 \le64\,\gamma(\cB_A(\Psi))^2<1,
 \qquad
 N_{\ind}^{\SEP}(\cB_A(\Psi))\ge3,
 \label{eq:C-AME-copy}
\end{equation}
for every $\operatorname{AME}(6,2)$ state and every choice of three-party subset $A$.

\suppsection{Moment-to-overlap uniformization}{app:uniformization}

We now prove Lemma~\ref{lem:moment-to-overlap}.  Retain its notation: $\Theta$ is a finite parameter set, $\{\ket{\Psi_\theta}:\theta\in\Theta\}$ are normalized states on $\cK=\bigotimes_{j=1}^n\C^{q_j}$, and the integer $m\ge1$ and constant $A_m\ge0$ satisfy the pointwise moment bound in Eq.~\eqref{eq:S-pointwise-family-moment}.  The purpose is to pass from that pointwise estimate to one state whose overlap is small for all normalized product states simultaneously.

For a finite-dimensional Hilbert space $\mathcal V$, write $\mathbb S(\mathcal V):=\{\ket v\in\mathcal V:\langle v|v\rangle=1\}$. Let $\mu_j$ be the unique unitarily invariant probability measure on $\mathbb S(\C^{q_j})$, and let $\mu_{\mathrm{prod}}:=\mu_1\otimes\cdots\otimes\mu_n$. We regard $\mu_{\mathrm{prod}}$ as a probability measure on the manifold of normalized product states
\begin{equation}
 \mathcal P_{\mathbf q}:=\left\{\ket\Phi=\bigotimes_{j=1}^n\ket{\phi_j}:\ket{\phi_j}\in\mathbb S(\C^{q_j})\right\}.
 \label{eq:B-product-manifold}
\end{equation}
For each $\theta\in\Theta$, define the product-Haar moment by the integral over this explicitly specified domain,
\begin{equation}
 R_m(\theta):=\int_{\mathcal P_{\mathbf q}}|\langle\Phi|\Psi_\theta\rangle|^{2m}\,d\mu_{\mathrm{prod}}(\Phi).
 \label{eq:B-Rm}
\end{equation}
Since the assumed pointwise bound holds for every $\ket\Phi\in\mathcal P_{\mathbf q}$ and $\mu_{\mathrm{prod}}$ is normalized, averaging gives
\begin{equation}
 \frac1{|\Theta|}\sum_{\theta\in\Theta}R_m(\theta)=\int_{\mathcal P_{\mathbf q}}\frac1{|\Theta|}\sum_{\theta\in\Theta}|\langle\Phi|\Psi_\theta\rangle|^{2m}\,d\mu_{\mathrm{prod}}(\Phi)\le A_m.
 \label{eq:B-average-R}
\end{equation}
The interchange of the finite sum and the integral is immediate; equivalently, it follows from linearity of integration.
All $R_m(\theta)$ are nonnegative.  Hence at least one parameter value $\theta_*$ obeys
\begin{equation}
 R_m(\theta_*)\le A_m.
 \label{eq:B-good-seed-average}
\end{equation}
The remaining task is to convert the Haar average $R_m(\theta_*)$ into a uniform maximum over normalized product states.

\subsection{The local Haar moment operator}
For a $q$-dimensional Hilbert space, let $\Pi_{\mathrm{sym}}^{(m)}$ be the orthogonal projector onto the symmetric subspace $\Sym^m(\C^q)\subset(\C^q)^{\otimes m}$ and set
\begin{equation}
 s_{m,q}:=\dim\Sym^m(\C^q)=\binom{q+m-1}{m}.
 \label{eq:B-sym-dim}
\end{equation}

\begin{lemma}[Local Haar moment operator, see also Ref. \cite{Mele2024Haar}]\label{lem:B-Haar}
Let $\mu_q$ be the unitarily invariant probability measure on $\mathbb S(\C^q)$. Then
\begin{equation}
 \int_{\mathbb S(\C^q)}(\proj{\phi})^{\otimes m}\,d\mu_q(\phi)=\frac{\Pi_{\mathrm{sym}}^{(m)}}{s_{m,q}}.
 \label{eq:Haar-sym}
\end{equation}
\end{lemma}

\begin{proof}
Denote the explicitly bounded integral on the left of Eq.~\eqref{eq:Haar-sym} by $T_m$.  Every vector $\ket{\phi}^{\otimes m}$ is invariant under permutations of its $m$ tensor factors, so $T_m$ is supported on $\Sym^m(\C^q)$; equivalently,
$T_m=\Pi_{\mathrm{sym}}^{(m)}T_m\Pi_{\mathrm{sym}}^{(m)}$.
Moreover, Haar invariance implies
\begin{equation}
 U^{\otimes m}T_m(U^\dagger)^{\otimes m}=T_m
 \label{eq:B-Haar-invariance}
\end{equation}
for every $U\in U(q)$, where $U(q)$ denotes the unitary group on $\C^q$.  The representation $U\mapsto U^{\otimes m}$ restricted to $\Sym^m(\C^q)$ is irreducible.  Schur's lemma therefore gives
$T_m=c\,\Pi_{\mathrm{sym}}^{(m)}$ for some scalar $c$.  Taking traces fixes the scalar: the integrand has unit trace, hence $\Tr T_m=1$, while $\Tr\Pi_{\mathrm{sym}}^{(m)}=s_{m,q}$.  Thus $c=s_{m,q}^{-1}$, proving Eq.~\eqref{eq:Haar-sym}.
\end{proof}

\subsection{From the Haar average to the largest product overlap}
For the multipartite space $\cK=\bigotimes_{j=1}^n\C^{q_j}$, write
\begin{equation}
 \Pi_j:=\Pi_{\mathrm{sym},j}^{(m)},
 \qquad
 s_j:=s_{m,q_j},
 \qquad
 \Pi:=\bigotimes_{j=1}^n\Pi_j.
 \label{eq:B-global-Pi}
\end{equation}
Here the $m$ replicas are regrouped by laboratory, so the ambient space is canonically identified as
\begin{equation}
 \cK^{\otimes m}
 \simeq
 \bigotimes_{j=1}^n(\C^{q_j})^{\otimes m}.
 \label{eq:B-replica-regroup}
\end{equation}
This regrouping is only a permutation of tensor factors and does not change inner products or norms.

Take an arbitrary normalized product state
$\ket{\Phi}=\bigotimes_{j=1}^n\ket{\phi_j}$.  Under the regrouping in Eq.~\eqref{eq:B-replica-regroup},
\begin{equation}
 \ket{\Phi}^{\otimes m}
 =\bigotimes_{j=1}^n\ket{\phi_j}^{\otimes m}.
 \label{eq:B-product-power}
\end{equation}
Each local tensor power $\ket{\phi_j}^{\otimes m}$ is completely symmetric under permutations of the $m$ replicas.  Hence
\begin{equation}
 \Pi_j\ket{\phi_j}^{\otimes m}=\ket{\phi_j}^{\otimes m}
 \quad\text{for every }j,
\end{equation}
and therefore
\begin{equation}
 \Pi\ket{\Phi}^{\otimes m}=\ket{\Phi}^{\otimes m}.
 \label{eq:B-Pi-fixes-product}
\end{equation}
Consequently the symmetric projector may be inserted into the overlap:
\begin{equation}
 |\langle\Phi|\Psi_\theta\rangle|^{2m}=|\langle\Phi^{\otimes m}|\Psi_\theta^{\otimes m}\rangle|^2=|\langle\Phi^{\otimes m}|\Pi|\Psi_\theta^{\otimes m}\rangle|^2.
 \label{eq:B-insert-Pi}
\end{equation}
The first equality follows from
$\langle\Phi^{\otimes m}|\Psi_\theta^{\otimes m}\rangle
 =\langle\Phi|\Psi_\theta\rangle^m$;
the second follows from Eq.~\eqref{eq:B-Pi-fixes-product} and the Hermiticity of $\Pi$.  Applying Cauchy--Schwarz and using $\|\Phi^{\otimes m}\|=1$ gives
\begin{equation}
 |\langle\Phi|\Psi_\theta\rangle|^{2m}\le\|\Phi^{\otimes m}\|^2\|\Pi\ket{\Psi_\theta}^{\otimes m}\|^2=\|\Pi\ket{\Psi_\theta}^{\otimes m}\|^2.
 \label{eq:maximum-from-proj}
\end{equation}
Thus the norm of the symmetric projection is an upper bound independent of the particular normalized product state $\ket{\Phi}$.

It remains to identify that norm with the Haar average in Eq.~\eqref{eq:B-Rm}.  Writing the overlap as a matrix element,
\begin{equation}
 R_m(\theta)=\int_{\mathcal P_{\mathbf q}}\langle\Psi_\theta^{\otimes m}|(\proj{\Phi})^{\otimes m}|\Psi_\theta^{\otimes m}\rangle\,d\mu_{\mathrm{prod}}(\Phi).
 \label{eq:B-R-operator}
\end{equation}
After regrouping the replicas by laboratory,
\begin{equation}
 (\proj{\Phi})^{\otimes m}
 =\bigotimes_{j=1}^n(\proj{\phi_j})^{\otimes m}.
 \label{eq:B-projector-factor}
\end{equation}
Because the Haar measure is a product measure, the integral factorizes.  Lemma~\ref{lem:B-Haar} therefore gives
\begin{equation}
 \int_{\mathcal P_{\mathbf q}}(\proj{\Phi})^{\otimes m}\,d\mu_{\mathrm{prod}}(\Phi)=\bigotimes_{j=1}^n\left[\int_{\mathbb S(\C^{q_j})}(\proj{\phi_j})^{\otimes m}\,d\mu_j(\phi_j)\right]=\frac{\Pi}{\prod_{j=1}^ns_j}.
 \label{eq:B-product-Haar-operator}
\end{equation}
Substitution into Eq.~\eqref{eq:B-R-operator} yields
\begin{equation}
 R_m(\theta)=\frac{\langle\Psi_\theta^{\otimes m}|\Pi|\Psi_\theta^{\otimes m}\rangle}{\prod_js_j}=\frac{\|\Pi\ket{\Psi_\theta}^{\otimes m}\|^2}{\prod_js_j}.
 \label{eq:B-R-projection}
\end{equation}
The second equality uses $\Pi=\Pi^\dagger=\Pi^2$.

Combining Eqs.~\eqref{eq:maximum-from-proj} and \eqref{eq:B-R-projection}, we obtain, for every normalized product state $\ket{\Phi}$,
\begin{equation}
 |\langle\Phi|\Psi_\theta\rangle|^{2m}
 \le R_m(\theta)\prod_{j=1}^ns_j.
 \label{eq:B-pointwise-from-R}
\end{equation}
Now choose the parameter $\theta_*$ from Eq.~\eqref{eq:B-good-seed-average}.  Since the right-hand side of Eq.~\eqref{eq:B-pointwise-from-R} no longer depends on $\ket{\Phi}$, we may maximize the left-hand side over all normalized product states and obtain
\begin{equation}
 \max_{\ket{\Phi}\,\mathrm{product}}
 |\langle\Phi|\Psi_{\theta_*}\rangle|^{2m}
 \le A_m\prod_{j=1}^n\binom{q_j+m-1}{m}.
 \label{eq:B-max-mth}
\end{equation}
Finally, the left-hand side is the $m$th power of the maximal squared product overlap.  Taking the $m$th root gives exactly Eq.~\eqref{eq:S-moment-to-overlap} and proves Lemma~\ref{lem:moment-to-overlap}.

\suppsection{Stabilizer moments and individual-copy complexity}{app:stabmoment}

We first review the $p$-ary stabilizer notation needed in both Theorems~\ref{thm:individual-unbounded} and~\ref{thm:stabilizer-collective}.  Let $p$ be an odd prime, write $\F_p$ for the field with $p$ elements, and set $\omega=e^{2\pi i/p}$.  Relative to the computational basis $\{\ket q:q\in\F_p\}$ of one $p$-level system, define the generalized shift and phase operators
\begin{equation}
 X_p\ket q=\ket{q+1},\qquad Z_p\ket q=\omega^q\ket q,
 \label{eq:D-XZ}
\end{equation}
where addition is modulo $p$.  For $(a,b)\in\F_p^2$, define the Weyl operator $W_p(a,b):=X_p^aZ_p^b$.  Directly from Eq.~\eqref{eq:D-XZ},
\begin{equation}
 W_p(a,b)W_p(c,d)=\omega^{bc-ad}W_p(c,d)W_p(a,b).
 \label{eq:D-weylcomm}
\end{equation}
The $n$-qudit generalized Pauli group is generated, up to phases, by tensor products $\bigotimes_{j=1}^nW_p(a_j,b_j)$.  A normalized state is called a \emph{stabilizer state} if it is a simultaneous eigenstate of a maximal Abelian subgroup of this Pauli group.  Such a subgroup is generated by $n$ independent commuting Pauli operators and has one-dimensional joint eigenspaces; its complete joint eigenbasis will be called a \emph{maximal-stabilizer eigenbasis}.  We denote by $\Stab_{n,p}$ the finite set of normalized stabilizer states on $(\C^p)^{\otimes n}$.

Fix an integer moment order $m\ge1$ with $n\ge m-1$.  The derivation uses the finite stabilizer-moment identity of Gross, Nezami, and Walter~\cite{GrossNezamiWalter2021}; all subsequent estimates are derived explicitly.

\subsection{Finite stabilizer moment identity}
On $\F_p^m\oplus\F_p^m$, write $\mathbf{x}\!\cdot\!\mathbf{y}:=\sum_{a=1}^m x_ay_a$ and define the split quadratic form $q$ together with its associated symmetric bilinear form $B$ by
\begin{equation}
 q(\mathbf{x},\mathbf{y})=\mathbf{x}\cdot\mathbf{x}-\mathbf{y}\cdot\mathbf{y},\qquad B((\mathbf{x},\mathbf{y}),(\mathbf{x}',\mathbf{y}'))=\mathbf{x}\cdot\mathbf{x}'-\mathbf{y}\cdot\mathbf{y}'.
 \label{eq:D-split-form}
\end{equation}
Let $\Sigma_{m,m}(p)$ be the set of $m$-dimensional subspaces $T\le\F_p^m\oplus\F_p^m$ that are totally isotropic for $B$, meaning $B(v,w)=0$ for all $v,w\in T$, and contain the stochastic vector $(\mathbf1,\mathbf1)$, where $\mathbf1=(1,\ldots,1)\in\F_p^m$.  For $\mathbf{x}=(x_1,\ldots,x_m)\in\F_p^m$, write $\ket{\mathbf{x}}=\bigotimes_{a=1}^m\ket{x_a}$ for the corresponding computational-basis vector in $(\C^p)^{\otimes m}$, and define
\begin{equation}
 r(T):=\sum_{(\mathbf{x},\mathbf{y})\in T}\ket{\mathbf{x}}\!\bra{\mathbf{y}},\qquad R_n(T):=r(T)^{\otimes n}.
 \label{eq:D-rT}
\end{equation}
Here $r(T)$ is standard notation for this one-site moment operator and is unrelated to the copy number $r$ in the discrimination problem. Via the Theorem 5.3 of  Ref. \cite{GrossNezamiWalter2021},  for $n\ge m-1$, the exact stabilizer moment identity is
\begin{equation}
 \frac1{|\Stab_{n,p}|}\sum_{\ket S\in\Stab_{n,p}}(\proj S)^{\otimes m}=\frac1{Z_{n,p,m}}\sum_{T\in\Sigma_{m,m}(p)}R_n(T),
 \label{eq:exact-stab-moment}
\end{equation}
where
\begin{equation}
 |\Sigma_{m,m}(p)|=\prod_{k=0}^{m-2}(p^k+1),\qquad Z_{n,p,m}=p^n\prod_{k=0}^{m-2}(p^k+p^n).
 \label{eq:Sigma-Z}
\end{equation}

\subsection{Tensor-power matrix-element bound}
We require a bound on the matrix element of $r(T)$ evaluated on a replica tensor power.  The only combinatorial input is the two-base case of the vector-matroid partition theorem~\cite{Edmonds1965}: a multiset of $2m$ vectors in an $m$-dimensional vector space can be partitioned into two bases if every submultiset $E'$ obeys $|E'|\le2\dim\operatorname{span}E'$.  We verify this condition directly from total isotropy.

\begin{lemma}[Product-power matrix-element bound]\label{lem:rT-contraction}
For every $T\in\Sigma_{m,m}(p)$ and every normalized state $\ket b\in\C^p$,
\begin{equation}
 |\langle b^{\otimes m}|r(T)|b^{\otimes m}\rangle|\le1.
 \label{eq:rT-contraction}
\end{equation}
\end{lemma}

\begin{proof}
Choose a linear isomorphism $L:\F_p^m\to T$ and write $L(\mathbf t)=(A\mathbf t,B\mathbf t)$.  Let $\ell_1,\ldots,\ell_{2m}\in(\F_p^m)^*$ be the coordinate linear forms obtained from the rows of $A$ and $B$.  Since $(\mathbf1,\mathbf1)\in T$, no coordinate vanishes identically on $T$, and hence none of the $\ell_i$ is the zero functional.

Let $V\le\F_p^m$, put $U=L(V)\le T$, and let $S(V)$ be the set of coordinate positions on which $U$ is not identically zero.  The restriction of the split form to the coordinate space supported on $S(V)$ is nondegenerate.  Since $U$ is totally isotropic in that space, $U\subseteq U^\perp$.  Nondegeneracy gives $\dim U+\dim U^\perp=|S(V)|$, and therefore
\begin{equation}
 2\dim V=2\dim U\le |S(V)|.
 \label{eq:support-rank}
\end{equation}
Now take a subspace $W\le(\F_p^m)^*$ and let $V=W^\perp:=\{\mathbf t\in\F_p^m:\ell(\mathbf t)=0\text{ for every }\ell\in W\}$ be its annihilator.  Every coordinate form $\ell_i\in W$ vanishes on $V$, so $i\notin S(V)$.  Therefore
\begin{equation}
 \#\{i:\ell_i\in W\}\le2m-|S(V)|\le2m-2\dim V=2\dim W.
 \label{eq:matroid-condition}
\end{equation}
For any submultiset $E'$ of the coordinate forms, choose $W=\operatorname{span}E'$.  Equation~\eqref{eq:matroid-condition} gives $|E'|\le2\dim\operatorname{span}E'$, so the partition theorem yields a disjoint decomposition $\{1,\ldots,2m\}=I\sqcup J$ for which both $\{\ell_i:i\in I\}$ and $\{\ell_i:i\in J\}$ are bases of $(\F_p^m)^*$.

Write $b_a=\langle a|b\rangle$.  Parameterizing $T$ by $L(\mathbf t)$ and taking absolute values gives
\begin{equation}
 |\langle b^{\otimes m}|r(T)|b^{\otimes m}\rangle|\le\sum_{\mathbf t\in\F_p^m}\left(\prod_{i\in I}|b_{\ell_i(\mathbf t)}|\right)\left(\prod_{i\in J}|b_{\ell_i(\mathbf t)}|\right).
 \label{eq:D-before-CS}
\end{equation}
Cauchy--Schwarz bounds the right-hand side by the product of
\begin{equation}
 \left[\sum_{\mathbf t\in\F_p^m}\prod_{i\in I}|b_{\ell_i(\mathbf t)}|^2\right]^{1/2},\qquad \left[\sum_{\mathbf t\in\F_p^m}\prod_{i\in J}|b_{\ell_i(\mathbf t)}|^2\right]^{1/2}.
 \label{eq:D-CS-factors}
\end{equation}
Because the $I$-forms constitute a basis of the dual space, $\mathbf t\mapsto(\ell_i(\mathbf t))_{i\in I}$ is a bijection of $\F_p^m$.  Hence
\begin{equation}
 \sum_{\mathbf t\in\F_p^m}\prod_{i\in I}|b_{\ell_i(\mathbf t)}|^2=\prod_{i\in I}\left(\sum_{a\in\F_p}|b_a|^2\right)=1,
 \label{eq:D-I-factor}
\end{equation}
and the same holds for $J$.  This proves Eq.~\eqref{eq:rT-contraction}.
\end{proof}

Only the matrix-element estimate in Lemma~\ref{lem:rT-contraction} is used below; no operator-norm bound $\|r(T)\|\le1$ is assumed.

\subsection{Product-state stabilizer moment}
Let $\ket\beta=\bigotimes_{j=1}^n\ket{b_j}$ be a normalized product state across the physical qudits.  Since $R_n(T)=r(T)^{\otimes n}$, Lemma~\ref{lem:rT-contraction} implies
\begin{equation}
 |\langle\beta^{\otimes m}|R_n(T)|\beta^{\otimes m}\rangle|=\prod_{j=1}^n|\langle b_j^{\otimes m}|r(T)|b_j^{\otimes m}\rangle|\le1.
 \label{eq:D-Rn-contraction}
\end{equation}
Taking the $\ket\beta^{\otimes m}$ matrix element of Eq.~\eqref{eq:exact-stab-moment}, applying the triangle inequality on the right, and using Eq.~\eqref{eq:D-Rn-contraction} gives
\begin{equation}
 \frac1{|\Stab_{n,p}|}\sum_{\ket S\in\Stab_{n,p}}|\langle\beta|S\rangle|^{2m}\le\frac{|\Sigma_{m,m}(p)|}{Z_{n,p,m}}=\frac{\prod_{k=0}^{m-2}(p^k+1)}{p^n\prod_{k=0}^{m-2}(p^k+p^n)}.
 \label{eq:D-exact-product-moment}
\end{equation}
For $0\le k\le m-2$, $p^k+1\le2p^k$ and $p^k+p^n\ge p^n$.  Therefore
\begin{equation}
 \frac{|\Sigma_{m,m}(p)|}{Z_{n,p,m}}\le2^{m-1}p^{-nm+(m-1)(m-2)/2},
 \label{eq:D-stab-moment-simple}
\end{equation}
which proves Eq.~\eqref{eq:stabilizer-Am-main}.

\subsection{Detailed asymptotic optimization and proof of Theorem~\ref{thm:individual-unbounded}}\label{app:thm1}
Apply Lemma~\ref{lem:moment-to-overlap} with $q_j=p$ to Eq.~\eqref{eq:D-stab-moment-simple}.  There exists a stabilizer state $\ket{S_*}$ such that
\begin{equation}
 \gamma(\ket{S_*})\le2^{1-1/m}p^{-n+(m-1)(m-2)/(2m)}\binom{p+m-1}{m}^{n/m}.
 \label{eq:D-gamma-direct}
\end{equation}
Write $\gamma(\ket{S_*})=p^{-n+\Delta_*}$, with $\Delta_*\ge0$ because $\gamma\ge p^{-n}$.  Taking base-$p$ logarithms yields
\begin{equation}
 \Delta_*\le \frac{n}{m}\log_p\binom{p+m-1}{m}+\frac{(m-1)(m-2)}{2m}+\left(1-\frac1m\right)\log_p2.
 \label{eq:D-Delta-bound}
\end{equation}
We determine the asymptotically optimal moment order. Since $p$ is fixed,
\begin{equation}
 \log_p\binom{p+m-1}{m}=\sum_{a=1}^{p-1}\log_p(m+a)-\log_p((p-1)!)=(p-1)\log_pm+O_p(1),
 \label{eq:binom-fixedp}
\end{equation}
and in fact the error is $-\log_p((p-1)!)+O_p(m^{-1})$.  Let $L_n:=\log_p n$ and set temporarily
\begin{equation}
 m=A\sqrt{nL_n},\qquad A>0 \text{ fixed}.
 \label{eq:D-m-scale}
\end{equation}
Then
\begin{equation}
 \log_pm=\frac12L_n+\frac12\log_pL_n+\log_pA,
 \label{eq:D-logm-expansion}
\end{equation}
so the first term in Eq.~\eqref{eq:D-Delta-bound} satisfies
\begin{equation}
 \frac{n}{m}\log_p\binom{p+m-1}{m}=\frac{p-1}{2A}\sqrt{nL_n}+O_p\!\left(\sqrt{\frac n{L_n}}\,\ln L_n\right).
 \label{eq:D-first-asymptotic}
\end{equation}
The second term has the exact expansion
\begin{equation}
 \frac{(m-1)(m-2)}{2m}=\frac m2-\frac32+\frac1m=\frac A2\sqrt{nL_n}+O(1),
 \label{eq:D-second-asymptotic}
\end{equation}
while the final term in Eq.~\eqref{eq:D-Delta-bound} is $O_p(1)$.  Consequently
\begin{equation}
 \Delta_*\le\left[\frac{p-1}{2A}+\frac A2+o(1)\right]\sqrt{n\log_p n}.
 \label{eq:D-A-bound}
\end{equation}
The function $f(A)=(p-1)/(2A)+A/2$ has $f'(A)=-(p-1)/(2A^2)+1/2$ and a unique minimum at $A=\sqrt{p-1}$, where $f(A)=\sqrt{p-1}$.  We therefore choose
\begin{equation}
 m_n=\left\lceil\sqrt{(p-1)n\log_p n}\right\rceil.
 \label{eq:D-m-choice}
\end{equation}
Rounding changes $m_n$ by $O(1)$ and hence only contributes $o(\sqrt{n\ln n})$ to Eq.~\eqref{eq:D-Delta-bound}.  Moreover $m_n=o(n)$, so the condition $n\ge m_n-1$ required by the moment identity holds for all sufficiently large $n$.  Thus
\begin{equation}
 \Delta_*\le(\sqrt{p-1}+o(1))\sqrt{n\log_p n}.
 \label{eq:Delta-opt}
\end{equation}

Choose a maximal stabilizer group for which $\ket{S_*}$ is a simultaneous eigenstate, and let $\cB_n$ be its complete joint eigenbasis. As  this basis lies in some a local-Weyl orbit, so every basis state is related to $\ket{S_*}$ by a tensor product of one-qudit Weyl unitaries, up to an overall phase.  Local-unitary invariance therefore gives $\gamma(\cB_n)=\gamma(\ket{S_*})=p^{-n+\Delta_*}$.  Proposition~\ref{prop:trace-overlap} yields
\begin{equation}
 P_{\ind}^{\SEP}(r;\cB_n)\le p^{-n+r\Delta_*}.
 \label{eq:stab-P-delta}
\end{equation}
If $\Delta_*=0$, the right-hand side is $p^{-n}<1$ for every finite $r$ and the desired lower bound is immediate.  Otherwise, every integer $r<n/\Delta_*$ is excluded from perfect discrimination.  Combining this with Eq.~\eqref{eq:Delta-opt} gives
\begin{equation}
 N_{\ind}^{\SEP}(\cB_n)\ge\left(\frac1{\sqrt{p-1}}-o(1)\right)\sqrt{\frac{n}{\log_p n}},
\end{equation}
which is Eq.~\eqref{eq:thm1-N}.  Finally, if $r=o(\sqrt{n/\log_p n})$, then $r\Delta_*=o(n)$ by Eq.~\eqref{eq:Delta-opt}, and Eq.~\eqref{eq:stab-P-delta} becomes
\begin{equation}
 P_{\ind}^{\SEP}(r;\cB_n)\le p^{-n+o(n)}=D^{-1+o(1)},
\end{equation}
proving Eq.~\eqref{eq:thm1-strong}.

\suppsection{Odd-prime stabilizer bases and collective readout}{app:stabilizerreadout}

We prove Theorem~\ref{thm:stabilizer-collective}, retaining the Weyl--Pauli conventions introduced in SM-D.  The aim is to express the protocol in notation natural for local state discrimination.  We reserve $r$ for the copy number in the performance quantities $P_{\alpha}^{\mathsf M}$ and $N_{\alpha}^{\mathsf M}$.

\subsection{Maximal-stabilizer eigenbases and local incompatibility}
Let $\cB_{\mathrm{stab}}=\{\ket{\psi_{\mathbf y}}:\mathbf y\in\F_p^n\}$ be a maximal-stabilizer eigenbasis.  Choose independent commuting Pauli generators $S_1,\ldots,S_n$ and label their one-dimensional joint eigenspaces by
\begin{equation}
 S_\ell\ket{\psi_{\mathbf y}}=\omega^{y_\ell}\ket{\psi_{\mathbf y}},\qquad \mathbf y=(y_1,\ldots,y_n)\in\F_p^n,\quad \ell=1,\ldots,n.
 \label{eq:stab-labels}
\end{equation}
Thus the unknown classical label is the stabilizer syndrome $\mathbf y$.

Factor each global generator as $S_\ell=\bigotimes_{j=1}^nS_\ell^{[j]}$.  Each local factor is a phase times a one-qudit Weyl operator, so write
\begin{equation}
 S_\ell^{[j]}=\omega^{\theta_{j\ell}}W_p(a_{j\ell},b_{j\ell}),\qquad \theta_{j\ell}\in\F_p,
 \label{eq:E-local-Weyl}
\end{equation}
and define the local commutation exponent
\begin{equation}
 \kappa_j(\ell,k):=b_{j\ell}a_{jk}-a_{j\ell}b_{jk}\in\F_p.
 \label{eq:E-kappa}
\end{equation}
By Eq.~\eqref{eq:D-weylcomm}, $S_\ell^{[j]}S_k^{[j]}=\omega^{\kappa_j(\ell,k)}S_k^{[j]}S_\ell^{[j]}$.  Global commutativity of $S_\ell$ and $S_k$ implies
\begin{equation}
 \sum_{j=1}^n\kappa_j(\ell,k)=0\pmod p,
 \label{eq:E-global-commutation}
\end{equation}
but the individual terms need not vanish.  Hence globally compatible stabilizer generators can remain locally incompatible at a given laboratory.  The collective protocol removes exactly these local commutation phases.

\subsection{Scalar-weight collective protocol}
First, we present a lemma that will be need below.
\begin{lemma}[Spectral projectors of the grouped Pauli observables]
\label{lem:E-spectral-projectors}
Let $R$ be a unitary operator satisfying $R^p=I$, where $p$ is
prime and $\omega=e^{2\pi i/p}$.  For $s\in\F_p$, define
\begin{equation}
 \Pi_R(s):=\frac1p\sum_{k=0}^{p-1}\omega^{-sk}R^k .
 \label{eq:E-cyclic-projector}
\end{equation}
Then $\Pi_R(s)$ is the orthogonal projector onto the eigenspace of
$R$ with eigenvalue $\omega^s$.  In particular,
\begin{equation}
 \Pi_R(s)\Pi_R(t)=\delta_{s,t}\Pi_R(s),
 \qquad
 \sum_{s\in\F_p}\Pi_R(s)=I .
 \label{eq:E-projector-relations}
\end{equation}
\end{lemma}

\begin{proof}
Since $R$ is unitary, it has an orthonormal eigenbasis.  If
$R|v\rangle=\lambda|v\rangle$, then $R^p=I$ implies
$\lambda^p=1$, so $\lambda=\omega^t$ for some $t\in\F_p$.
Therefore
\begin{equation}
 \Pi_R(s)|v\rangle
 =
 \frac1p\sum_{k=0}^{p-1}\omega^{(t-s)k}|v\rangle
 =
 \delta_{s,t}|v\rangle,
 \label{eq:E-projector-action}
\end{equation}
where we used the finite character identity
\[
 \frac1p\sum_{k=0}^{p-1}\omega^{ak}
 =
 \begin{cases}
  1,&a=0\pmod p,\\
  0,&a\ne0\pmod p.
 \end{cases}
\]
Hence $\Pi_R(s)$ acts as the identity on the $\omega^s$ eigenspace
and as zero on every other eigenspace.  It is therefore the
orthogonal spectral projector of $R$ associated with $\omega^s$.
The relations in Eq.~\eqref{eq:E-projector-relations} follow
immediately.
\end{proof}

We now give a sufficient condition under which several copies cancel all local Pauli incompatibilities simultaneously while retaining the stabilizer-eigenvalue information needed to reconstruct the unknown basis label.  If $c\in\F_p$ and $A$ is one of the Pauli unitaries used below, $A^c$ denotes $A^{\widetilde c}$ for an integer representative $\widetilde c\in\{0,\ldots,p-1\}$.  This is independent of the representative because, for odd prime $p$, these Pauli unitaries have order dividing $p$.

From Eq.~\eqref{eq:E-kappa}, for $u,v\in\F_p$,
\begin{equation}
 \bigl(S_\ell^{[j]}\bigr)^u\bigl(S_k^{[j]}\bigr)^v=\omega^{uv\kappa_j(\ell,k)}\bigl(S_k^{[j]}\bigr)^v\bigl(S_\ell^{[j]}\bigr)^u .
 \label{eq:E-powered-comm}
\end{equation}
For nonnegative integer representatives this follows by repeatedly commuting one factor through the other; periodicity modulo $p$ then gives the stated relation for field elements.

\begin{proposition}[Scalar-weight stabilizer decoder]\label{prop:E-weighted-decoder}
Let $c_1,\ldots,c_r\in\F_p$ satisfy
\begin{equation}
 \sum_{a=1}^r c_a^2=0,\qquad c_*:=\sum_{a=1}^r c_a\ne0 \quad\text{in }\F_p.
 \label{eq:scalarcriterion-main}
\end{equation}
Then every maximal-stabilizer eigenbasis of $n$ qudits in local dimension $p$ is perfectly distinguishable from $r$ copies by one-round collective LOCC.
\end{proposition}

\begin{proof}
Let the unknown basis state be $\ket{\psi_{\mathbf y}}$, where $\mathbf y=(y_1,\ldots,y_n)\in\F_p^n$, and let the repeated input be $\ket{\Psi_{\mathbf y}}:=\ket{\psi_{\mathbf y}}^{\otimes r}$.

\paragraph*{Step 1: commuting observables at each laboratory.}
For laboratory $j$ and stabilizer generator $\ell$, define
\begin{equation}
 R_{j\ell}:=\bigotimes_{a=1}^r\left[\bigl(S_\ell^{[j]}\bigr)^{c_a}\right]^{(a)},
 \label{eq:Rjl}
\end{equation}
where $(a)$ indicates copy $a$.  Fix $j$ and two generators $\ell,k$.  Operators acting on different copies commute, and Eq.~\eqref{eq:E-powered-comm} gives
\begin{equation}
 R_{j\ell}R_{jk}=\omega^{\kappa_j(\ell,k)\sum_{a=1}^r c_a^2}R_{jk}R_{j\ell}=R_{jk}R_{j\ell}.
 \label{eq:E-local-commuting}
\end{equation}
Thus $R_{j1},\ldots,R_{jn}$ commute pairwise for every laboratory $j$.  The same weights $c_1,\ldots,c_r$ cancel all local commutation phases because the condition $\sum_{a=1}^r c_a^2=0$ is independent of $j,\ell,$ and $k$.

\paragraph*{Step 2: local collective measurement.}
Each $R_{j\ell}$ is unitary and satisfies $R_{j\ell}^p=I$, so its eigenvalues belong to $\{1,\omega,\ldots,\omega^{p-1}\}$.  For $s\in\F_p$, define
\begin{equation}
 \Pi_{j\ell}(s):=\frac1p\sum_{k=0}^{p-1}\omega^{-sk}R_{j\ell}^{\,k}.
 \label{eq:E-single-spectral-projector}
\end{equation}
This is the orthogonal spectral projector of $R_{j\ell}$ associated with eigenvalue $\omega^s$.  Indeed, if $R_{j\ell}\ket v=\omega^a\ket v$, then finite character orthogonality gives $\Pi_{j\ell}(s)\ket v=p^{-1}\sum_{k=0}^{p-1}\omega^{(a-s)k}\ket v=\delta_{a,s}\ket v$.  Hence $\Pi_{j\ell}(s)\Pi_{j\ell}(s')=\delta_{s,s'}\Pi_{j\ell}(s)$ and $\sum_{s\in\F_p}\Pi_{j\ell}(s)=I$.  Since the operators $R_{j1},\ldots,R_{jn}$ commute, their spectral projectors commute as well.  For $\mathbf s_j=(s_{j1},\ldots,s_{jn})\in\F_p^n$, define the joint spectral projector
\begin{equation}
 \Pi_j(\mathbf s_j):=\prod_{\ell=1}^n\Pi_{j\ell}(s_{j\ell}).
 \label{eq:E-joint-spectral-projector}
\end{equation}
The product is well defined because the $R_{j\ell}$ commute.  The nonzero projectors $\{\Pi_j(\mathbf s_j)\}_{\mathbf s_j\in\F_p^n}$ form a projective measurement on the $r$ systems held by laboratory $j$; explicitly,
\begin{equation}
 \Pi_j(\mathbf s_j)\Pi_j(\mathbf s'_j)=0\ \ (\mathbf s_j\ne\mathbf s'_j),\qquad \sum_{\mathbf s_j\in\F_p^n}\Pi_j(\mathbf s_j)=I.
 \label{eq:E-local-PVM}
\end{equation}
Some projectors may vanish and some may have rank greater than one.  If outcome $\mathbf s_j$ occurs, then
\begin{equation}
 R_{j\ell}\Pi_j(\mathbf s_j)=\omega^{s_{j\ell}}\Pi_j(\mathbf s_j),\qquad \ell=1,\ldots,n.
 \label{eq:E-local-outcome-eigen}
\end{equation}
Hence one collective local measurement returns the full vector $\mathbf s_j$.  After all laboratories perform these fixed measurements, they broadcast the classical outcomes.  A global outcome is $\mathbf s=(\mathbf s_1,\ldots,\mathbf s_n)$ with projector $\Pi_{\mathbf s}:=\bigotimes_{j=1}^n\Pi_j(\mathbf s_j)$.

\paragraph*{Step 3: reconstruction of the stabilizer syndrome.}
For each generator $\ell$, define
\begin{equation}
 R_\ell:=\bigotimes_{j=1}^nR_{j\ell}=\bigotimes_{a=1}^r\bigl(S_\ell^{c_a}\bigr)^{(a)}.
 \label{eq:E-global-R}
\end{equation}
The second equality is the canonical regrouping of tensor factors and uses $S_\ell=\bigotimes_jS_\ell^{[j]}$.  Equation~\eqref{eq:stab-labels} gives
\begin{equation}
 R_\ell\ket{\Psi_{\mathbf y}}=\omega^{c_*y_\ell}\ket{\Psi_{\mathbf y}}.
 \label{eq:E-global-syndrome}
\end{equation}
The global outcome projector is an eigenspace projector of the same observable:
\begin{equation}
 R_\ell\Pi_{\mathbf s}=\omega^{\sum_{j=1}^n s_{j\ell}}\Pi_{\mathbf s}.
 \label{eq:E-R-on-outcome}
\end{equation}
Suppose that the outcome $\mathbf s$ occurs with nonzero probability on $\ket{\Psi_{\mathbf y}}$.  Since $\Pi_{\mathbf s}$ is a projector, $\Pi_{\mathbf s}\ket{\Psi_{\mathbf y}}\ne0$.  Moreover $\Pi_{\mathbf s}$ commutes with every $R_\ell$.  Therefore Eqs.~\eqref{eq:E-global-syndrome} and~\eqref{eq:E-R-on-outcome} imply
\begin{equation}
 \omega^{\sum_j s_{j\ell}}\Pi_{\mathbf s}\ket{\Psi_{\mathbf y}}=R_\ell\Pi_{\mathbf s}\ket{\Psi_{\mathbf y}}=\Pi_{\mathbf s}R_\ell\ket{\Psi_{\mathbf y}}=\omega^{c_*y_\ell}\Pi_{\mathbf s}\ket{\Psi_{\mathbf y}}.
 \label{eq:E-eigenvalue-comparison}
\end{equation}
Because the projected vector is nonzero and $\omega$ is a primitive $p$th root of unity,
\begin{equation}
 \sum_{j=1}^n s_{j\ell}=c_*y_\ell\pmod p,
 \qquad \ell=1,\ldots,n.
 \label{eq:E-syndrome-equation}
\end{equation}
Since $c_*\ne0$ in $\F_p$, it is invertible and the label is recovered componentwise as
\begin{equation}
 y_\ell=c_*^{-1}\sum_{j=1}^n s_{j\ell}\pmod p,
 \qquad \ell=1,\ldots,n.
 \label{eq:E-syndrome-recovery}
\end{equation}
Thus every nonzero-probability global outcome is compatible with exactly one basis label.  Each laboratory performs one fixed projective measurement on its $r$ local copies, followed by a single broadcast of classical data.  Hence the protocol is one-round collective LOCC.
\end{proof}

\subsection{Number of copies}

It remains to show that the scalar weights in Eq.~\eqref{eq:scalarcriterion-main} always exist with $r\le3$.

\begin{proposition}[Two or three scalar weights]\label{prop:E-number-weights}
For every odd prime $p$, Eq.~\eqref{eq:scalarcriterion-main} has a solution with
\begin{equation}
 \tau_p=\begin{cases}2,&p\equiv1\pmod4,\\3,&p\equiv3\pmod4.\end{cases}
 \label{eq:E-taup}
\end{equation}
Moreover, when $p\equiv3\pmod4$, no two-copy choice satisfies the scalar-weight criterion.
\end{proposition}

\begin{proof}
First take $t=2$. If $c_1^2+c_2^2=0$ and $c_1+c_2\ne0$, then neither coefficient vanishes and $(c_1c_2^{-1})^2=-1$. Conversely, if $\lambda^2=-1$ in $\F_p$, then $(c_1,c_2)=(1,\lambda)$ obeys $c_1^2+c_2^2=0$ and $1+\lambda\ne0$. By Euler's criterion, $-1$ is a quadratic residue in $\F_p$ exactly when $p\equiv1\pmod4$. Hence two scalar-weighted copies suffice in this case, while no two-copy scalar-weight choice exists for $p\equiv3\pmod4$.

Now assume $p\equiv3\pmod4$. Let
\begin{equation}
 Q:=\{x^2:x\in\F_p\}
 \label{eq:E-square-set}
\end{equation}
be the set of quadratic residues together with $0$.  Since every nonzero square has exactly two square roots and $0$ has one, $|Q|=(p+1)/2$.  The translated set $-1-Q:=\{-1-q:q\in Q\}$ has the same cardinality.  Because
\begin{equation}
 |Q|+|-1-Q|=p+1>p=|\F_p|,
 \label{eq:E-pigeonhole}
\end{equation}
the two sets intersect.  Hence there exist $a,b\in\F_p$ such that $a^2=-1-b^2$, or equivalently
\begin{equation}
 a^2+b^2=-1.
 \label{eq:E-threeweight-equation}
\end{equation}
If $a+b+1\ne0$, then $(c_1,c_2,c_3)=(a,b,1)$ already satisfies Eq.~\eqref{eq:scalarcriterion-main}.  If instead $a+b+1=0$, then $a\ne0$: otherwise $b=-1$, and Eq.~\eqref{eq:E-threeweight-equation} would give $1=-1$, impossible for odd $p$.  Replacing $a$ by $-a$ preserves Eq.~\eqref{eq:E-threeweight-equation}, while
\begin{equation}
 (-a)+b+1=-2a\ne0.
 \label{eq:E-sign-flip}
\end{equation}
Thus in all cases there is a pair $(a,b)$ with $a^2+b^2=-1$ and $a+b+1\ne0$.  Taking $(c_1,c_2,c_3)=(a,b,1)$ proves the scalar-weight criterion with three copies.  Proposition~\ref{prop:E-weighted-decoder} therefore gives a three-copy one-round collective-LOCC protocol, completing the proof of Theorem~\ref{thm:stabilizer-collective}.
\end{proof}

\suppsection{Accessible information from identification probability}{app:information}

We now prove the information bound used in Eq.~\eqref{eq:info-vs-P-main}.  The only additional ingredient beyond the definitions in the Letter is that all measurement classes considered here are closed under classical postprocessing of their outcomes.

\begin{lemma}[Accessible information versus identification probability]
\label{lem:information-identification}
Let $\cB=\{\ket{\psi_y}\}_{y=1}^{D}$ be a complete orthonormal basis, with the label $Y$ uniformly distributed over $\{1,\ldots,D\}$.  For every $r\ge1$, $\mathsf M\in\{\LOCC,\SEP\}$, and $\alpha\in\{\ind,\col\}$,
\begin{equation}
 I_{\alpha}^{\mathsf M}(r;\cB)
 \le
 \log_2D+\log_2P_{\alpha}^{\mathsf M}(r;\cB).
 \label{eq:F-information-identification}
\end{equation}
\end{lemma}

\begin{proof}
Fix an arbitrary measurement
$$ \mathcal{M}:=
 \{M_z\}\in\mathsf M_{\alpha}^{(r)}
$$
and denote its classical outcome by $Z_{\mathcal{M}}$ and denote it $Z$ for simplicity.  On input label $y$, the conditional probability of obtaining outcome $z$ is
\begin{equation}
 p(z|y)
 =
 \Tr\!\left(M_z\rho_y^{\otimes r}\right).
 \label{eq:F-channel}
\end{equation}
Since $Y$ is uniform,
\begin{equation}
 p(y,z)
 =
 \frac1D p(z|y),
 \qquad
 p(z)
 =
 \frac1D\sum_{y=1}^{D}p(z|y).
 \label{eq:F-joint}
\end{equation}

We first relate the conditional entropy $H(Y|Z)$ to the optimal probability of guessing $Y$ from the observed value of $Z$.  For every $z$ with $p(z)>0$, define
\begin{equation}
 m_z:=\max_y p(y|z).
 \label{eq:F-mz}
\end{equation}
For every $y$ such that $p(y|z)>0$,
$
 p(y|z)\le m_z,
$
and hence
\[
 -\log_2p(y|z)\ge-\log_2m_z.
\]
Using the convention $0\log_2 0=0$, it follows that
\begin{align}
 H(Y|Z=z)
 &=
 -\sum_y p(y|z)\log_2p(y|z) \ge
 -\sum_y p(y|z)\log_2m_z =
 -\log_2m_z.
 \label{eq:F-pointwise-entropy}
\end{align}
Averaging over the measurement outcome gives
\begin{align}
 H(Y|Z)
 &=
 \sum_zp(z)H(Y|Z=z)
  \ge
 \sum_zp(z)\bigl[-\log_2m_z\bigr].
 \label{eq:F-average-entropy}
\end{align}
Since $x\mapsto-\log_2x$ is convex on $(0,\infty)$, Jensen's inequality yields
\begin{align}
 H(Y|Z)
 &\ge
 -\log_2\!\left(\sum_zp(z)m_z\right)
  =
 -\log_2\!\left[
 \sum_zp(z)\max_y p(y|z)
 \right].
 \label{eq:F-Jensen}
\end{align}
Define
\begin{equation}
 p_g(Y|Z)
 :=
 \sum_zp(z)\max_y p(y|z).
 \label{eq:F-pg}
\end{equation}
This is precisely the optimal average probability of correctly guessing $Y$ after observing $Z$.  Equation~\eqref{eq:F-Jensen} therefore gives
\begin{equation}
 H(Y|Z)\ge-\log_2p_g(Y|Z).
 \label{eq:F-Shannon-guess}
\end{equation}

It remains to compare $p_g(Y|Z)$ with the discrimination probability
$P_{\alpha}^{\mathsf M}(r;\cB)$ defined in Eq.~\eqref{eq:Pdef}.  For each outcome $z$, choose a maximum-likelihood decision rule
\begin{equation}
 \widehat y(z)\in\arg\max_y p(y|z).
 \label{eq:F-ML-rule}
\end{equation}
If the maximum is attained by more than one label, choose any one of them.  Coarse-graining all measurement outcomes that lead to the same decision, define
\begin{equation}
 N_y
 :=
 \sum_{z:\,\widehat y(z)=y}M_z,
 \qquad y=1,\ldots,D.
 \label{eq:F-coarse-graining}
\end{equation}
The operators $\{N_y\}_{y=1}^{D}$ form a POVM, since
\[
 N_y\ge0,
 \qquad
 \sum_yN_y=\sum_zM_z=I.
\]
Moreover, $\{N_y\}$ belongs to the same measurement class
$\mathsf M_{\alpha}^{(r)}$.  Indeed, Eq.~\eqref{eq:F-coarse-graining} only applies a classical relabeling and coarse graining to the final measurement outcome; it introduces no additional quantum operation and hence preserves each of the LOCC and SEP classes, in both the individual-copy and collective architectures.

The success probability of this $D$-outcome discrimination measurement is
\begin{align}
 P_{\rm succ}(\{N_y\})
 &=
 \frac1D\sum_y
 \Tr\!\left(N_y\rho_y^{\otimes r}\right)
 \nonumber\\
 &=
 \frac1D\sum_z
 \Tr\!\left(
 M_z\rho_{\widehat y(z)}^{\otimes r}
 \right)
 \nonumber\\
 &=
 \sum_z
 \max_y
 \frac1D
 \Tr\!\left(M_z\rho_y^{\otimes r}\right)
 \nonumber\\
 &=
 \sum_z\max_y p(y,z)
 \nonumber\\
 &=
 \sum_zp(z)\max_y p(y|z)
 \nonumber\\
 &=
 p_g(Y|Z).
 \label{eq:F-success-pg}
\end{align}
Since $P_{\alpha}^{\mathsf M}(r;\cB)$ is the supremum of the discrimination success probability over all allowed $D$-outcome measurements in $\mathsf M_{\alpha}^{(r)}$, Eq.~\eqref{eq:F-success-pg} implies
\begin{equation}
 p_g(Y|Z)
 \le
 P_{\alpha}^{\mathsf M}(r;\cB).
 \label{eq:F-pg-P}
\end{equation}

Finally, uniformity of the basis label gives
\[
 H(Y)=\log_2D.
\]
Combining Eqs.~\eqref{eq:F-Shannon-guess} and~\eqref{eq:F-pg-P},
\begin{align}
 I(Y{:}Z)
 &=
 H(Y)-H(Y|Z) \le
 \log_2D+\log_2p_g(Y|Z) 
  \le
 \log_2D+
 \log_2P_{\alpha}^{\mathsf M}(r;\cB).
 \label{eq:F-info-bound}
\end{align}
This holds for every allowed measurement $\{M_z\}\in\mathsf M_{\alpha}^{(r)}$.  Taking the supremum over such measurements and using the definition of $I_{\alpha}^{\mathsf M}(r;\cB)$ proves Eq.~\eqref{eq:F-information-identification}, and hence Eq.~\eqref{eq:info-vs-P-main}.
\end{proof}

For the $p\equiv1\pmod4$ stabilizer family, Theorem~\ref{thm:stabilizer-collective} gives a two-copy collective-LOCC decoder, and therefore
\[
 I_{\col}^{\LOCC}(2;\cB_n)=\log_2D.
\]
In the subthreshold regime of Theorem~\ref{thm:individual-unbounded},
\[
 P_{\ind}^{\SEP}(r;\cB_n)=D^{-1+o(1)}.
\]
Equation~\eqref{eq:F-information-identification} then yields
\[
 I_{\ind}^{\SEP}(r;\cB_n)=o(\log_2D),
\]
which proves Corollary~\ref{cor:information-separation}.

\suppsection{Flat local-unitary orbit bases and Newton-identity replica rigidity}{app:flat}

We now work in arbitrary fixed physical local dimension $d\ge2$.  Write $\mathbb Z_d:=\mathbb Z/d\mathbb Z$ for the physical alphabet and put $D=d^n$.  No finite-field structure is assumed for $\mathbb Z_d$.

\subsection{Flat amplitudes give a complete local-unitary orbit}
For an arbitrary phase function $f:\mathbb Z_d^n\to\mathbb R$, define $\ket{\Omega_f}$ as in Eq.~\eqref{eq:flat-general}.  Thus a global computational label is a vector $\mathbf{x}=(x_1,\ldots,x_n)\in\mathbb Z_d^n$ and $\ket{\mathbf{x}}=\bigotimes_j\ket{x_j}$.  Set $\zeta:=e^{2\pi i/d}$ and let the one-site phase operator be $Z_d\ket a=\zeta^a\ket a$.  For $\mathbf{y}=(y_1,\ldots,y_n)\in\mathbb Z_d^n$, define
\begin{equation}
 Z_d^{\mathbf{y}}:=\bigotimes_{j=1}^n Z_d^{y_j},\qquad \mathbf{y}\!\cdot\!\mathbf{x}:=\sum_{j=1}^ny_jx_j\pmod d,
 \label{eq:G-Ztensor}
\end{equation}
so that $Z_d^{\mathbf{y}}\ket{\mathbf{x}}=\zeta^{\mathbf{y}\cdot\mathbf{x}}\ket{\mathbf{x}}$. For $\mathbf{y},\mathbf{z}\in\mathbb Z_d^n$,
\begin{equation}
 \langle\Omega_f|Z_d^{\mathbf{z}-\mathbf{y}}|\Omega_f\rangle=\frac1D\sum_{\mathbf{x}\in\mathbb Z_d^n}\zeta^{(\mathbf{z}-\mathbf{y})\cdot\mathbf{x}}=\prod_{j=1}^n\left[\frac1d\sum_{x_j=0}^{d-1}\zeta^{(z_j-y_j)x_j}\right]=\delta_{\mathbf{y},\mathbf{z}}.
 \label{eq:G-flat-orthogonality}
\end{equation}
The phase function cancels between the two amplitudes. Each one-dimensional character sum equals one when $z_j=y_j$ and zero otherwise. Thus the $D$ states $Z_d^{\mathbf{y}}\ket{\Omega_f}$ are mutually orthonormal and form a complete basis, proving Eq.~\eqref{eq:flat-orbit}. Moreover, $(Z_d^{\mathbf{y}})^{\otimes r}=\bigotimes_j(Z_d^{y_j})^{\otimes r}$ is local with respect to the grouped partition, so local-unitary invariance gives the same $\Gamma_r$ for every basis state.

\subsection{Replica rigidity}
Fix an integer $h\ge1$ and choose a prime $P>h$; write $\F_P$ for the field with $P$ elements. Define the moment-curve vector
\begin{equation}
 \mathbf{g}(\alpha)=(\alpha,\alpha^2,\ldots,\alpha^h)^T\in\F_P^h.
 \label{eq:G-moment-curve}
\end{equation}

\begin{lemma}[Replica rigidity]\label{lem:replica-rigidity}
If
\begin{equation}
 \sum_{a=1}^h \mathbf{g}(\alpha_a)=\sum_{a=1}^h \mathbf{g}(\beta_a),
 \label{eq:G-rigidity-assumption}
\end{equation}
with repetitions allowed, then the two multisets
$\{\alpha_1,\ldots,\alpha_h\}_{\mathrm{multi}}$ and
$\{\beta_1,\ldots,\beta_h\}_{\mathrm{multi}}$ are equal.
\end{lemma}

\begin{proof}
Set $\boldsymbol{\alpha}:=(\alpha_1,\ldots,\alpha_h)$ and $\boldsymbol{\beta}:=(\beta_1,\ldots,\beta_h)$. Equality of the $s$th coordinates in Eq.~\eqref{eq:G-rigidity-assumption} gives equality of the first $h$ power sums,
\begin{equation}
 P_s(\boldsymbol{\alpha}):=\sum_{a=1}^h\alpha_a^s=\sum_{a=1}^h\beta_a^s=:P_s(\boldsymbol{\beta}),
 \qquad s=1,\ldots,h.
 \label{eq:G-power-sums}
\end{equation}
Let $e_s$ denote the elementary symmetric polynomial of degree $s$, with $e_0=1$.  Newton's identities are
\begin{equation}
 s e_s
 =\sum_{j=1}^s(-1)^{j-1}e_{s-j}P_j,
 \qquad s=1,\ldots,h.
 \label{eq:G-Newton}
\end{equation}
Because $P>h$, every integer $1,\ldots,h$ is nonzero and hence invertible in $\F_P$.  Equation~\eqref{eq:G-Newton} therefore determines $e_1$ from $P_1$, then $e_2$ from $P_1,P_2$, and recursively all $e_s$ from the first $s$ power sums.  The equalities in Eq.~\eqref{eq:G-power-sums} imply that the two lists have identical elementary symmetric polynomials.  Consequently they define the same monic polynomial,
\begin{equation}
 \prod_{a=1}^h(t-\alpha_a)
 =t^h-e_1t^{h-1}+\cdots+(-1)^he_h
 =\prod_{a=1}^h(t-\beta_a).
 \label{eq:G-root-poly}
\end{equation}
Equality of these factorizations over a field implies equality of the roots with multiplicities, which is exactly equality of the two multisets.
\end{proof}

\suppsection{Moment-curve replica moments and proof of Theorem~\ref{thm:collective-unbounded}}{app:phasemoment}

\subsection{Exact replica-moment estimate}
Fix $r,m\ge1$ and put $h=mr$. Choose a prime $P>\max\{D,h\}$ and assign distinct nonzero field elements $\alpha_{\mathbf{x}}\in\F_P$ to all global labels $\mathbf{x}\in\mathbb Z_d^n$. Put $\mathbf{g}_{\mathbf{x}}:=\mathbf{g}(\alpha_{\mathbf{x}})=(\alpha_{\mathbf{x}},\alpha_{\mathbf{x}}^2,\ldots,\alpha_{\mathbf{x}}^h)^T\in\F_P^h$. For a phase parameter $\mathbf{u}=(u_1,\ldots,u_h)\in\F_P^h$, let $\mathbf{u}\!\cdot\!\mathbf{g}_{\mathbf{x}}$ denote the standard dot product over $\F_P$, set $\omega_P:=e^{2\pi i/P}$, and define
\begin{equation}
 \ket{\Omega_{\mathbf{u}}}:=D^{-1/2}\sum_{\mathbf{x}\in\mathbb Z_d^n}\omega_P^{\mathbf{u}\cdot\mathbf{g}_{\mathbf{x}}}\ket{\mathbf{x}}.
 \label{eq:momentcurve-state}
\end{equation}
The auxiliary prime $P$ is used only to choose phases; it is unrelated to the physical local dimension $d$. Throughout SM-\ref{app:phasemoment}, whenever the parameter space is $\F_P^k$, $\mathbf{u}$ is uniform on that finite set and
\begin{equation}
 \mathbb E_{\mathbf{u}}F(\mathbf{u}):=P^{-k}\sum_{\mathbf{u}\in\F_P^k}F(\mathbf{u}),
 \label{eq:H-expectation-convention}
\end{equation}
Thus every expectation over phase parameters is a finite average over an explicitly specified parameter space.

Let
\begin{equation}
 \ket\Phi=\bigotimes_{j=1}^n\ket{\phi_j},
 \qquad
 \ket{\phi_j}\in(\C^d)^{\otimes r},
 \label{eq:H-grouped-detector}
\end{equation}
be an arbitrary normalized product state with respect to the grouped laboratory partition.  Each local factor may be entangled across its $r$ copies.

For an $r$-tuple of global computational strings $\boldsymbol{\xi}=(\mathbf{x}^{(1)},\ldots,\mathbf{x}^{(r)})\in(\mathbb Z_d^n)^r$, write $\ket{\boldsymbol{\xi}}:=\bigotimes_{t=1}^r\ket{\mathbf{x}^{(t)}}$ and define
\begin{equation}
 a_{\boldsymbol{\xi}}:=\langle\Phi|\boldsymbol{\xi}\rangle.
 \label{eq:H-a-def}
\end{equation}
Because the vectors $\ket{\boldsymbol{\xi}}$ form an orthonormal basis of the $r$-copy global Hilbert space and $\ket\Phi$ is normalized, $\sum_{\boldsymbol{\xi}}|a_{\boldsymbol{\xi}}|^2=1$. Expanding $\ket{\Omega_{\mathbf{u}}}^{\otimes r}$ gives
\begin{equation}
 X_{\mathbf{u}}(\ket\Phi):=\langle\Phi|\Omega_{\mathbf{u}}^{\otimes r}\rangle=D^{-r/2}\sum_{\boldsymbol{\xi}}a_{\boldsymbol{\xi}}\,\omega_P^{\mathbf{u}\cdot\sum_{t=1}^r\mathbf{g}_{\mathbf{x}^{(t)}}}.
 \label{eq:Xu}
\end{equation}

We now compute its $2m$th moment over the uniformly distributed phase parameter $\mathbf{u}$. Expanding $X_{\mathbf{u}}^m\overline{X_{\mathbf{u}}}^m$ produces $m$ ordered $r$-tuples on each side, hence $h=mr$ global labels on each side. Flatten the $m$ positive-side blocks into the ordered sequence $\mathbf{X}=(\mathbf{x}_1,\ldots,\mathbf{x}_h)$ and define
\begin{equation}
 b_{\mathbf{X}}:=\prod_{s=1}^m a_{(\mathbf{x}_{(s-1)r+1},\ldots,\mathbf{x}_{sr})}.
 \label{eq:H-bX}
\end{equation}
Define $b_{\mathbf{Y}}$ analogously for the negative-side sequence $\mathbf{Y}=(\mathbf{y}_1,\ldots,\mathbf{y}_h)$.  Then
\begin{equation}
 |X_{\mathbf{u}}(\ket\Phi)|^{2m}=D^{-h}\sum_{\mathbf{X},\mathbf{Y}}b_{\mathbf{X}}\overline{b_{\mathbf{Y}}}\,\omega_P^{\mathbf{u}\cdot(\sum_{a=1}^h\mathbf{g}_{\mathbf{x}_a}-\sum_{a=1}^h\mathbf{g}_{\mathbf{y}_a})}.
 \label{eq:H-expanded-moment}
\end{equation}
Average uniformly over $\mathbf{u}\in\F_P^h$. Character orthogonality gives, for every $\mathbf{w}\in\F_P^h$,
\begin{equation}
 \frac1{P^h}\sum_{\mathbf{u}\in\F_P^h}\omega_P^{\mathbf{u}\cdot\mathbf{w}}=\begin{cases}1,&\mathbf{w}=\mathbf{0},\\0,&\mathbf{w}\ne\mathbf{0}.\end{cases}
 \label{eq:H-character-orthogonality}
\end{equation}
Thus a pair $(\mathbf{X},\mathbf{Y})$ survives the average only when
\begin{equation}
 \sum_{a=1}^h\mathbf{g}_{\mathbf{x}_a}=\sum_{a=1}^h\mathbf{g}_{\mathbf{y}_a}.
 \label{eq:H-survival}
\end{equation}
Because the map $\mathbf{x}\mapsto\alpha_{\mathbf{x}}$ is injective, Lemma~\ref{lem:replica-rigidity} implies that $\mathbf{X}$ and $\mathbf{Y}$ contain exactly the same multiset of global labels. Equivalently, there exists a permutation $\pi\in S_h$, where $S_h$ is the symmetric group on $h$ symbols, such that $\mathbf y_a=\mathbf x_{\pi(a)}$ for every $a$; when labels repeat, more than one permutation may represent the same ordered sequence.

For a multiset $M$ of $h$ global labels, let $F_M$ be the set of its distinct orderings. Regrouping all surviving terms by $M$ gives
\begin{equation}
 \mathbb E_{\mathbf{u}}|X_{\mathbf{u}}(\ket\Phi)|^{2m}=D^{-h}\sum_M\left|\sum_{\mathbf{X}\in F_M}b_{\mathbf{X}}\right|^2.
 \label{eq:H-multiset-sum}
\end{equation}
For each multiset, Cauchy--Schwarz yields
\begin{equation}
 \left|\sum_{\mathbf{X}\in F_M}b_{\mathbf{X}}\right|^2\le |F_M|\sum_{\mathbf{X}\in F_M}|b_{\mathbf{X}}|^2\le h!\sum_{\mathbf{X}\in F_M}|b_{\mathbf{X}}|^2.
 \label{eq:H-CS-multiset}
\end{equation}
The sets $F_M$ partition all ordered sequences $\mathbf{X}$, so summing Eq.~\eqref{eq:H-CS-multiset} over $M$ gives
\begin{equation}
 \mathbb E_{\mathbf{u}}|X_{\mathbf{u}}(\ket\Phi)|^{2m}\le\frac{h!}{D^h}\sum_{\mathbf{X}}|b_{\mathbf{X}}|^2.
 \label{eq:H-before-normalization}
\end{equation}
The block definition in Eq.~\eqref{eq:H-bX} factorizes the last sum:
\begin{equation}
 \sum_{\mathbf{X}}|b_{\mathbf{X}}|^2=\left(\sum_{\boldsymbol{\xi}}|a_{\boldsymbol{\xi}}|^2\right)^m=1.
 \label{eq:H-b-normalization}
\end{equation}  Therefore
\begin{equation}
 \mathbb E_{\mathbf{u}}|\langle\Phi|\Omega_{\mathbf{u}}^{\otimes r}\rangle|^{2m}
 \le\frac{h!}{D^h}
 =\frac{(mr)!}{D^{mr}},
 \label{eq:H-phase-moment}
\end{equation}
which is Eq.~\eqref{eq:phase-Am-main}.

Lemma~\ref{lem:moment-to-overlap}, applied to the grouped local dimensions $q_j=d^r$, now gives a state $\ket{\Omega_{\mathbf{u}_*}}$ satisfying
\begin{equation}
 \Gamma_r(\ket{\Omega_{\mathbf{u}_*}})\le\left[\frac{(mr)!}{D^{mr}}\binom{d^r+m-1}{m}^{n}\right]^{1/m}=D^{-r}\binom{d^r+m-1}{m}^{n/m}(mr)!^{1/m}.
 \label{eq:Gamma-flat-main}
\end{equation}
Since $\cA_{\mathbf{u}_*,d}$ is the local-unitary orbit of $\ket{\Omega_{\mathbf{u}_*}}$, $\Gamma_r(\cA_{\mathbf{u}_*,d})=\Gamma_r(\ket{\Omega_{\mathbf{u}_*}})$.  Proposition~\ref{prop:trace-overlap} then gives Eq.~\eqref{eq:flat-master-main}.

For convenience, define
\begin{equation}
 B_{n,r}^{(d)}(m)
 :=\frac1{d^n}\binom{d^r+m-1}{m}^{n/m}(mr)!^{1/m}.
 \label{eq:Bmaster}
\end{equation}
For each fixed $r$ and $m$, Eqs.~\eqref{eq:Gamma-flat-main} and~\eqref{eq:flat-master-main} show that there exists a phase parameter for which the collective-SEP success probability is at most $B_{n,r}^{(d)}(m)$.  We first optimize this quantity for a fixed block size and then select a common phase parameter for the block sizes required in Theorem~\ref{thm:collective-unbounded}.

\subsection{Detailed asymptotic optimization for a fixed block size}
The auxiliary prime can always be chosen larger than $\max\{D,mr\}$, so it imposes no restriction on the asymptotic choice of the moment order.  Put
\begin{equation}
 q:=d^r,
\end{equation}
and, for an integer moment order $m$, set $t:=m/q$.  We first minimize the resulting upper bound as a continuous function of $t>0$ and then round the corresponding value of $m=tq$ to an integer.  Define
\begin{equation}
 g(t):=(1+t)\ln(1+t)-t\ln t.
 \label{eq:H-g}
\end{equation}
For $N=q+m-1$ and $K=m$, set $x:=K/N$.  The binomial theorem gives
\begin{equation}
 1=(x+1-x)^N\ge\binom NK x^K(1-x)^{N-K},
\end{equation}
so, after taking logarithms,
$\ln\binom NK\le Nh(x)$ with $h(x):=-x\ln x-(1-x)\ln(1-x)$; the endpoint cases $K=0,N$ are immediate.  Here
\begin{equation}
 Nh(K/N)=(q+m-1)\ln(q+m-1)-m\ln m-(q-1)\ln(q-1).
\end{equation}
Substituting $m=tq$ into this expression and applying the mean-value theorem to the shifts by one gives
\begin{equation}
 \ln\binom{q+m-1}{m}\le qg(t)+O(\ln q+\ln(1+t)),
 \label{eq:binom-entropy}
\end{equation}
where the implicit constant is absolute, i.e., independent of $n,r,d,$ and $t$.  Stirling's formula, in the form $\ln N!=N\ln N-N+O(\ln N)$, gives
\begin{equation}
 \frac1m\ln(mr)!=r\ln(mr)-r+O\!\left(\frac{\ln(mr)}m\right).
 \label{eq:stirling}
\end{equation}
Starting from
\begin{equation}
 \ln B_{n,r}^{(d)}(m)=-n\ln d+\frac{n}{m}\ln\binom{q+m-1}{m}+\frac1m\ln(mr)!,
 \label{eq:H-logB-start}
\end{equation}
and using $m=tq$ and $\ln q=r\ln d$, we obtain
\begin{equation}
 \ln B_{n,r}^{(d)}(m)\le-n\ln d+r^2\ln d+r\ln r+r\ln t-r+\frac ntg(t)+\mathcal E_{n,r}(t),
 \label{eq:logB-detailed}
\end{equation}
where
\begin{equation}
 \mathcal E_{n,r}(t)=O\!\left(\frac{n[\ln q+\ln(1+t)]}{tq}\right)+O\!\left(\frac{\ln(trq)}{tq}\right).
 \label{eq:H-remainder}
\end{equation}
The $t$-dependent leading part of Eq.~\eqref{eq:logB-detailed} is
\begin{equation}
 F_{n,r}(t):=r\ln t+\frac ntg(t).
\end{equation}
Using $tg'(t)-g(t)=-\ln(1+t)$,
\begin{equation}
 F_{n,r}'(t)=\frac{rt-n\ln(1+t)}{t^2}.
\end{equation}
Thus the stationary equation is
\begin{equation}
 rt=n\ln(1+t).
 \label{eq:tstationary}
\end{equation}
The function $t/\ln(1+t)$ is strictly increasing on $(0,\infty)$ because $\ln(1+t)>t/(1+t)$.  Hence Eq.~\eqref{eq:tstationary} has a unique positive solution whenever $n/r>1$, and this solution is the unique minimum of $F_{n,r}$.  In the remainder of this subsection, $t=t_{n,r}$ denotes this minimizing solution.

Set $L:=n/r$.  Equation~\eqref{eq:tstationary} becomes $t=L\ln(1+t)$.  Writing $t=Ls$ gives
\begin{equation}
 s=\ln L+\ln(s+L^{-1}).
 \label{eq:H-s-identity}
\end{equation}
As $L\to\infty$, the unique solution has $t\to\infty$ and hence $s=\ln(1+t)\to\infty$.  For all sufficiently large $L$, Eq.~\eqref{eq:H-s-identity} and $s\ge1$ imply
$ s\le\ln L+\ln(2s)$.  Since $\ln(2s)\le s/2$ for all sufficiently large $s$, this gives $s\le2\ln L$ eventually and therefore $s=O(\ln L)$.  Returning to Eq.~\eqref{eq:H-s-identity}, we obtain $\ln(s+L^{-1})=O(\ln\ln L)$ and hence
$s=\ln L+O(\ln\ln L)$.  Consequently $s/\ln L=1+O(\ln\ln L/\ln L)$, so
$\ln(s+L^{-1})=\ln\ln L+o(1)$.  Substitution in Eq.~\eqref{eq:H-s-identity} yields
\begin{equation}
 s=\ln L+\ln\ln L+o(1),
\end{equation}
and therefore
\begin{equation}
 t=\frac nr\left[\ln\frac nr+\ln\ln\frac nr+o(1)\right].
 \label{eq:t-asympt}
\end{equation}
At this minimizing solution, for $1\le r\le\sqrt n$, $L=n/r\ge\sqrt n$, and the stationary equation gives $n/t=r/\ln(1+t)$.  Hence the first term in Eq.~\eqref{eq:H-remainder} is
\begin{equation}
 O\!\left(\frac{r}{q}\left[1+\frac{\ln q}{\ln(1+t)}\right]\right)=O_d(r),
 \label{eq:H-remainder-first}
\end{equation}
and it is $o(r)$ whenever $r\to\infty$, because $q=d^r$.  The second term in Eq.~\eqref{eq:H-remainder} is $o(1)$ uniformly in the same range.  Thus $\mathcal E_{n,r}(t)=O_d(r)$ uniformly for $1\le r\le\sqrt n$, and $\mathcal E_{n,r}(t)=o(r)$ when $r\to\infty$.

At the stationary point,
\begin{equation}
 \frac ntg(t)=\frac nt\ln(1+t)+n\ln(1+1/t)=r+n\ln(1+1/t).
\end{equation}
The first term cancels the $-r$ in Eq.~\eqref{eq:logB-detailed}.  Moreover, since $t\to\infty$,
\begin{equation}
 n\ln(1+1/t)=\frac nt+O\!\left(\frac n{t^2}\right)=O\!\left(\frac r{\ln(n/r)}\right)=o(r).
\end{equation}
Also Eq.~\eqref{eq:t-asympt} gives
\begin{equation}
 \ln t=\ln\frac nr+\ln\ln\frac nr+o(1).
\end{equation}
Hence
\begin{equation}
 r\ln r+r\ln t=r\ln n+r\ln\ln\frac nr+o(r).
\end{equation}
Substitution into Eq.~\eqref{eq:logB-detailed} therefore gives, after division by $\ln d$,
\begin{equation}
 -\log_d B_{n,r}^{(d)}(m)\ge n-r^2-r\log_dn-r\log_d\ln(n/r)-O_d(r).
 \label{eq:optimized-B}
\end{equation}
Let $m_*=tq$ denote the real optimizer of the continuous upper bound and choose an integer $\widehat m$ with $|\widehat m-m_*|\le1$.  It remains to verify that this rounding does not change the retained asymptotic order.  For consecutive integers, writing $C_m:=\binom{q+m-1}{m}$ and using $C_{m+1}/C_m=(q+m)/(m+1)$ gives
\begin{equation}
 \left|\frac{n}{m+1}\ln C_{m+1}-\frac{n}{m}\ln C_m\right|\le\frac{n\ln C_m}{m(m+1)}+\frac{n}{m+1}\ln\frac{q+m}{m+1}.
 \label{eq:H-round-binomial}
\end{equation}
For the consecutive integers adjacent to $m_*=tq$, one has $m\asymp tq$, and Eq.~\eqref{eq:binom-entropy} gives $\ln C_m=O(q\ln(1+t))$, while $\ln[(q+m)/(m+1)]=O(1/t)$.  The right-hand side of Eq.~\eqref{eq:H-round-binomial} is therefore $o(r)$ at the scale in Eq.~\eqref{eq:t-asympt}.  Similarly,
\begin{equation}
 \left|\frac{\ln((m+1)r)!}{m+1}-\frac{\ln(mr)!}{m}\right|\le\frac{\ln(mr)!}{m(m+1)}+\frac1{m+1}\ln\frac{((m+1)r)!}{(mr)!}=O\!\left(\frac{r\ln(mr)}m\right)=o(r).
 \label{eq:H-round-factorial}
\end{equation}
Thus replacing $m_*$ by $\widehat m$ changes $\ln B_{n,r}^{(d)}$ by $o(r)$, and Eq.~\eqref{eq:optimized-B} remains valid for the integer moment order used in the construction.

For $r=\alpha\sqrt n$ with fixed $0<\alpha<1$, the two logarithmic correction terms in Eq.~\eqref{eq:optimized-B} are $O_d(\sqrt n\ln n)$, and therefore
\begin{equation}
 -\log_d B_{n,r}^{(d)}(m)\ge(1-\alpha^2)n-O_d(\sqrt n\ln n).
 \label{eq:alpha-B}
\end{equation}

\subsection{A common moment-curve phase choice}

To prove the strong-converse statement for a single sequence of bases, the phase parameter must satisfy the required moment estimates simultaneously for all block sizes in the relevant range. Fix an integer $R\le\sqrt n$.  For each $1\le r\le R$, choose an integer $m_r$ by rounding the optimizer obtained above and set
\begin{equation}
 h_r:=m_rr,\qquad H:=\max_{1\le r\le R}h_r.
 \label{eq:H-common-H}
\end{equation}
Choose a prime $P>\max\{D,H\}$ and assign distinct nonzero elements $\alpha_{\mathbf{x}}\in\F_P$ to the $D$ global labels $\mathbf{x}\in\mathbb Z_d^n$.  Use the common $H$-coordinate moment curve
\begin{equation}
 \mathbf g_{\mathbf{x}}^{(H)}:=(\alpha_{\mathbf{x}},\alpha_{\mathbf{x}}^2,\ldots,\alpha_{\mathbf{x}}^H)^T\in\F_P^H.
 \label{eq:H-common-curve}
\end{equation}
For $\mathbf u\in\F_P^H$, define the corresponding flat state
\begin{equation}
 \ket{\Omega_{\mathbf u}^{(H)}}:=D^{-1/2}\sum_{\mathbf{x}\in\mathbb Z_d^n}\omega_P^{\mathbf u\cdot\mathbf g_{\mathbf{x}}^{(H)}}\ket{\mathbf{x}},
 \qquad \omega_P=e^{2\pi i/P}.
 \label{eq:H-common-state}
\end{equation}
The parameter $\mathbf u$ is taken uniformly from the finite set $\F_P^H$, so every expectation below is the normalized finite sum $P^{-H}\sum_{\mathbf u\in\F_P^H}$.

For each $r$, let
\begin{equation}
 \mathcal P_r:=\left\{\ket\Phi=\bigotimes_{j=1}^n\ket{\phi_j}:\ket{\phi_j}\in\mathbb S((\C^d)^{\otimes r})\right\},
 \label{eq:H-Pr}
\end{equation}
and let $\mu_r$ be the product of the normalized unitarily invariant probability measures on the $n$ local unit spheres.  Define the nonnegative product-Haar moment
\begin{equation}
 \mathcal R_r(\mathbf u):=\int_{\mathcal P_r}|\langle\Phi|(\Omega_{\mathbf u}^{(H)})^{\otimes r}\rangle|^{2m_r}\,d\mu_r(\Phi).
 \label{eq:H-Rr}
\end{equation}
For a fixed $r$, the phase average in the expansion of $\mathcal R_r$ is taken with the $H$-coordinate vectors in Eq.~\eqref{eq:H-common-curve}.  A surviving term must therefore satisfy equality of all $H$ coordinate sums, and in particular equality of the first $h_r$ power sums.  Since $P>H\ge h_r$, Lemma~\ref{lem:replica-rigidity}, applied to these first $h_r$ coordinates, implies equality of the two multisets of $h_r$ global labels.  The calculation leading to Eq.~\eqref{eq:H-phase-moment} thus applies without change and gives
\begin{equation}
 \mathbb E_{\mathbf u}\mathcal R_r(\mathbf u)\le \Lambda_r,
 \qquad
 \Lambda_r:=\frac{(m_rr)!}{D^{m_rr}},
 \qquad 1\le r\le R.
 \label{eq:H-common-average}
\end{equation}

Now combine all block sizes into one nonnegative functional,
\begin{equation}
 \mathcal F(\mathbf u):=\frac1R\sum_{r=1}^R\frac{\mathcal R_r(\mathbf u)}{\Lambda_r}.
 \label{eq:H-aggregate-functional}
\end{equation}
Using Eq.~\eqref{eq:H-common-average} and linearity of the finite average,
\begin{equation}
 \mathbb E_{\mathbf u}\mathcal F(\mathbf u)=\frac1R\sum_{r=1}^R\frac{\mathbb E_{\mathbf u}\mathcal R_r(\mathbf u)}{\Lambda_r}\le1.
 \label{eq:H-aggregate-average}
\end{equation}
Since $\mathcal F$ is nonnegative and its finite average is at most one, there exists at least one phase parameter $\mathbf u_*$ for which $\mathcal F(\mathbf u_*)\le1$.  Since every summand in Eq.~\eqref{eq:H-aggregate-functional} is nonnegative, this single choice obeys, simultaneously for every $1\le r\le R$,
\begin{equation}
 \mathcal R_r(\mathbf u_*)\le R\Lambda_r=R\frac{(m_rr)!}{D^{m_rr}}.
 \label{eq:H-simultaneous-R}
\end{equation}
Applying the Haar-to-maximum step from SM-C to Eq.~\eqref{eq:H-simultaneous-R}, with grouped local dimension $d^r$, yields
\begin{equation}
 \Gamma_r(\ket{\Omega_{\mathbf u_*}^{(H)}})\le\left[R\frac{(m_rr)!}{D^{m_rr}}\binom{d^r+m_r-1}{m_r}^{n}\right]^{1/m_r}.
 \label{eq:H-simultaneous-Gamma}
\end{equation}
The orbit basis $\cA_{\mathbf u_*,d}:=\{Z_d^{\mathbf y}\ket{\Omega_{\mathbf u_*}^{(H)}}:\mathbf y\in\mathbb Z_d^n\}$ has the same grouped product overlap for every basis state.  Proposition~\ref{prop:trace-overlap} therefore gives the simultaneous discrimination bound
\begin{equation}
 P_{\col}^{\SEP}(r;\cA_{\mathbf u_*,d})\le R^{1/m_r}B_{n,r}^{(d)}(m_r),
 \qquad 1\le r\le R.
 \label{eq:simultaneous-bound}
\end{equation}
For $R\le\sqrt n$, the optimizer in Eq.~\eqref{eq:t-asympt} satisfies $m_r\asymp d^r(n/r)\ln(n/r)$ uniformly for $1\le r\le R$.  In particular $m_r\ge c_d\sqrt n\ln n$ for some constant $c_d>0$ depending only on $d$ and all sufficiently large $n$, so
\begin{equation}
 \sup_{1\le r\le R}\frac{\ln R}{m_r}=o(1),\qquad R^{1/m_r}=1+o(1)
 \label{eq:H-R-factor}
\end{equation}
uniformly over the entire range.

\subsection{Completion of Theorem~\ref{thm:collective-unbounded}}

We first prove the copy-complexity lower bound.  By the meaning of the term $O_d(r)$ in Eq.~\eqref{eq:optimized-B}, there is a constant $K_d>0$, depending only on $d$, such that for all sufficiently large $n$ and all $1\le r\le\sqrt n$ the optimized integer moment order satisfies
\begin{equation}
 -\log_d B_{n,r}^{(d)}(m_r)\ge n-r^2-r\log_dn-r\log_d\ln(n/r)-K_dr.
 \label{eq:H-explicit-Kd}
\end{equation}
Let
\begin{equation}
 a_n:=\frac12\log_dn+\frac12\log_d\ln n+C_d,
 \qquad
 r_n^*:=\sqrt n-a_n,
 \qquad
 R_n:=\lfloor r_n^*\rfloor,
 \label{eq:H-Rn}
\end{equation}
where the constant $C_d>0$, depending only on $d$, will be fixed below.  Since $a_n=O_d(\ln n)=o(\sqrt n)$, $r_n^*>0$ for all sufficiently large $n$.  Define the right-hand side of Eq.~\eqref{eq:H-explicit-Kd} as a function of a real variable $r$,
\begin{equation}
 E_{n,d}(r):=n-r^2-r\log_dn-r\log_d\ln(n/r)-K_dr.
 \label{eq:H-E-function}
\end{equation}
For $1\le r\le\sqrt n$,
\begin{equation}
 E_{n,d}'(r)=-2r-\log_dn-K_d-\log_d\ln(n/r)+\frac{1}{\ln d\,\ln(n/r)}<0
 \label{eq:H-E-derivative}
\end{equation}
for all sufficiently large $n$.  Thus $E_{n,d}$ is decreasing on the relevant interval and, because $R_n\le r_n^*$,
\begin{equation}
 E_{n,d}(R_n)\ge E_{n,d}(r_n^*).
 \label{eq:H-floor-monotonicity}
\end{equation}
It remains to evaluate the continuous point $r_n^*$.  Directly,
\begin{equation}
 n-(r_n^*)^2=2a_n\sqrt n-a_n^2,
 \qquad
 r_n^*\log_dn=\sqrt n\,\log_dn-a_n\log_dn.
 \label{eq:H-first-two-expand}
\end{equation}
Furthermore,
\begin{equation}
 \ln\frac{n}{r_n^*}=\frac12\ln n-\ln\!\left(1-\frac{a_n}{\sqrt n}\right)=\frac12\ln n+O_d\!\left(\frac{\ln n}{\sqrt n}\right),
 \label{eq:H-log-n-over-r}
\end{equation}
so
\begin{equation}
 \log_d\ln\frac{n}{r_n^*}=\log_d\ln n-\log_d2+O_d\!\left(\frac{1}{\sqrt n}\right).
 \label{eq:H-loglog-expand}
\end{equation}
Consequently,
\begin{equation}
 r_n^*\log_d\ln\frac{n}{r_n^*}=\sqrt n\,[\log_d\ln n-\log_d2]+O_d(\ln n\,\ln\ln n).
 \label{eq:H-third-expand}
\end{equation}
Substituting Eqs.~\eqref{eq:H-first-two-expand} and~\eqref{eq:H-third-expand} into Eq.~\eqref{eq:H-E-function}, and using the definition of $a_n$, gives
\begin{equation}
 E_{n,d}(r_n^*)=[2C_d+\log_d2-K_d]\sqrt n+O_d((\ln n)^2).
 \label{eq:H-endpoint-exponent-continuous}
\end{equation}
Choose $C_d$ so that $2C_d+\log_d2-K_d>1$.  Since $(\ln n)^2=o(\sqrt n)$, Eqs.~\eqref{eq:H-explicit-Kd}, \eqref{eq:H-floor-monotonicity}, and~\eqref{eq:H-endpoint-exponent-continuous} imply
\begin{equation}
 -\log_d B_{n,R_n}^{(d)}(m_{R_n})\ge E_{n,d}(R_n)>0
 \label{eq:H-endpoint-exponent}
\end{equation}
for all sufficiently large $n$.

Apply the common phase construction with $R=R_n$ and denote the resulting orbit basis by $\cA_{n,d}$.  From Eq.~\eqref{eq:H-R-factor},
\begin{equation}
 \log_d R_n^{1/m_{R_n}}=\frac{\ln R_n}{m_{R_n}\ln d}=o(1).
 \label{eq:H-prefactor-log}
\end{equation}
Equations~\eqref{eq:simultaneous-bound}, \eqref{eq:H-endpoint-exponent}, and~\eqref{eq:H-prefactor-log} therefore imply
\begin{equation}
 P_{\col}^{\SEP}(R_n;\cA_{n,d})<1
 \label{eq:H-endpoint-below-one}
\end{equation}
for all sufficiently large $n$.  Lemma~\ref{lem:A-copy-monotonicity} then excludes perfect collective-SEP discrimination with every $r\le R_n$, so
\begin{equation}
 N_{\col}^{\SEP}(\cA_{n,d})>R_n=\sqrt n-\frac12\log_dn-\frac12\log_d\ln n-O_d(1),
\end{equation}
which proves the first assertion of Theorem~\ref{thm:collective-unbounded}.

For the strong-converse statement, fix $0<\alpha<1$ and set $r_n:=\lfloor\alpha\sqrt n\rfloor$.  Since $R_n/\sqrt n\to1$, $r_n\le R_n$ for all sufficiently large $n$.  Equation~\eqref{eq:alpha-B} gives
\begin{equation}
 -\log_d B_{n,r_n}^{(d)}(m_{r_n})\ge(1-\alpha^2)n-O_d(\sqrt n\ln n)=(1-\alpha^2)n+o(n).
 \label{eq:H-alpha-explicit}
\end{equation}
The factor $R_n^{1/m_{r_n}}$ contributes only $o(1)$ to the base-$d$ logarithm by Eq.~\eqref{eq:H-R-factor}.  Hence Eq.~\eqref{eq:simultaneous-bound} yields
\begin{equation}
 P_{\col}^{\SEP}(r_n;\cA_{n,d})\le d^{-(1-\alpha^2)n+o(n)}=D^{-(1-\alpha^2)+o(1)}.
 \label{eq:H-strong-converse}
\end{equation}
Finally, the information bound proved in SM-\ref{app:information} gives
\begin{equation}
 I_{\col}^{\SEP}(r_n;\cA_{n,d})\le(\alpha^2+o(1))\log_2D.
\end{equation}
This proves the remaining assertions of Theorem~\ref{thm:collective-unbounded}.

\end{document}